\documentclass[draft, fontsize=11pt, headings=normal, headinclude=true, paper=a4, pagesize, cleardoublepage=current, BCOR=10mm, fleqn, numbers=noenddot]{scrartcl}
\pdfoutput=1
\usepackage{lmodern}

\setkomafont{sectioning}{\normalcolor\bfseries}
\recalctypearea 
\setcapindent{0pt}

  \usepackage[english]{babel}  
  \usepackage[numbers,square]{natbib}
  \usepackage{mathrsfs}
  \usepackage{amsmath}
  \usepackage{amssymb}
  \usepackage{amsthm}

  \usepackage[utf8]{inputenc}    
  \usepackage[T1]{fontenc}         
  \usepackage{graphicx}            
  \usepackage{color}		
  \usepackage[english=british]{csquotes} 
  \usepackage{enumerate}

\newtheorem{thm}{Theorem}[section]
\newtheorem{prop}[thm]{Proposition}
\newtheorem{lem}[thm]{Lemma}

\newtheorem*{thm*}{Theorem}

\theoremstyle{definition}
\newtheorem{definition}[thm]{Definition}
\newtheorem{assumption}[thm]{Assumption}

\newtheorem{rem}[thm]{Remark}

\renewcommand{\phi}{\varphi}
\newcommand{\eps}{\varepsilon}

\newcommand{\ud}{\mathrm{d}}

\newcommand{\ue}{\mathrm{e}}
\newcommand{\ui}{\mathrm{i}}
\newcommand{\R}{\mathbb{R}}
\newcommand{\C}{\mathbb{C}}
\newcommand{\N}{\mathbb{N}}
\newcommand{\cF}{\mathcal{F}}
\newcommand{\cN}{\mathcal{N}}
\newcommand{\cK}{\mathcal{K}}

\newcommand{\uppar}[1]{\ensuremath{^{(#1)}}}

\title{Hamiltonians for polaron models with subcritical ultraviolet singularities}

\author{Jonas Lampart
\thanks{CNRS \& LICB (UMR 6303), Université de Bourgogne Franche-Comté, 9 Av. A. Savary, 21078 Dijon Cedex, France.
\texttt{lampart@math.cnrs.fr}}
}

\begin{document}

\maketitle

\begin{abstract}
We treat the ultraviolet problem for polaron-type models in nonrelativistic quantum field theory. Assuming that the dispersion relations of particles and the field have the same growth at infinity, we cover all subcritical (superrenormalisable) interactions.
The Hamiltonian without cutoff is exhibited as an explicit self-adjoint operator obtained by a finite iteration procedure. The cutoff Hamiltonians converge to this operator in the strong resolvent sense after subtraction of a perturbative approximation for the ground-state energy.
\end{abstract}

\tableofcontents

\section{Introduction} 

In this article we consider models for a particle interacting with a bosonic quantum field by a linear coupling. Such nonrelativistic quantum field theory (QFT) models play an important role in mathematical physics, on the one hand as approximations for more fundamental relativistic QFTs, and on the other hand as effective models that feature interactions with a field of quasi-particles. Examples of such models include the Nelson model~\cite{nelson1964, GrWu18, LaSch19}, the optical Fröhlich model for an electron and phonons in a solid~\cite{Hfrohlich1954, GrWu16}, and the Bogoliubov-Fröhlich model for an impurity interacting with excitations of a Bose-Einstein condensate~\cite{grusdt2016, La20}. 

We will consider translation-invariant models and fix the total momentum, the sum of the momenta of the particle and the field. The formal expression for the polaron Hamiltonian of one particle interacting with the field at total momentum $P$ is given on the momentum Fock space $\cF$ by
\begin{equation}\label{eq:H formal}
H= \Omega(\ud \Gamma(k) - P) + \ud \Gamma(\omega) + a(v)+a^*(v),
\end{equation}
where $\Omega$, $\omega$ are the dispersion relations of the particle and the field, respectively, $a,a^*$ denote  bosonic annihilation/creation operators and $\ud \Gamma(A)$ acts as the one-particle operator $A$ on each boson (see Section~\ref{sect:notation} for a summary of the notation). This expression, taken literally, defines an operator with dense domain only if $v$ is an element of the one-particle space $L^2(\R^d)$, as otherwise $a(v)$ is not closable and $a(v)^*$ is not densely defined.

The goal of this article is to make sense of the operator~\eqref{eq:H formal} in situations presenting an ultraviolet (UV) singularity, that is $v\notin L^2$ due to the behaviour of $v\in L^2_\mathrm{loc}$ for large momentum. Although we only treat here the model with one particle, we expect that our results can be generalised to any fixed number of particles interacting with the field.
This problem has been studied for several classes of $\Omega$, $\omega$,  $v$ in the literature~\cite{nelson1964, eckmann1970,gross1973, GuHiLo14, GrWu16, GrWu18, wuensch17, LaSch19, IBCrelat, IBCmassless}. The novelty of our result is the introduction of a new general and explicit method allowing for stronger singularities than previous results, up to the threshold of critical singularities that we explain below. This method generalises the one used in~\cite{La19, La20} for the specific case of the Bogoliubov-Fröhlich Hamiltonian.
The underlying idea is that one might be able to make sense of the expression~\eqref{eq:H formal} if one can find a domain $D$ such that the action of the individual terms in $H$ on $\Psi\in D$ may not yield an element of $\cF$, but, due to cancellations between the different terms, their sum is an element of $\cF$. The conditions on $\Psi$ leading to such cancellations are known as (abstract) interior boundary conditions and have recently been studied for a variety of models~\cite{thomas1984, yafaev1992, TeTu20, KeSi16, IBCpaper, LaSch19, La19,schmidt2019, IBCrelat, IBCmassless, tumulka2020, henheik2020, lienert2020, La2021polaron, BiLa21}. 

\subsection{Ultraviolet scaling}

The strength of the UV singularity depends on the behaviour of the functions $\Omega$, $\omega$,  $v$ (from $\R^d$ to $\R$) as their argument tends to infinity. This can be analysed heuristically if the functions are (essentially) homogeneous. Assume for the moment that $\Omega(p)=|p|^\gamma$, $\omega(k)=|k|^\beta$, $v(k)=g|k|^{-\alpha}$ are exactly homogeneous with exponents $\beta, \gamma \geq 0$, $\alpha\in \R$, and coupling strength $g>0$. 
Let $U_\lambda\psi(k) =\lambda^{d/2} \psi(\lambda k)$ be the unitary rescaling on $L^2(\R^d)$. Conjugating each term in~\eqref{eq:H formal}
with the lift $\Gamma(U_\lambda)$ of this unitary to Fock space (acting as $U_\lambda$ on each tensor factor), we obtain 
\begin{equation}\label{eq:H formal resc}
 \Gamma(U_\lambda)H \Gamma(U_\lambda)^*=\lambda^\gamma \Omega(\ud \Gamma(k) - \lambda^{-1}P) + \lambda^\beta \ud \Gamma(\omega) + \lambda^{d/2-\alpha}(a(v)+a^*(v)),
\end{equation}
which is of a similar form as~\eqref{eq:H formal} but with modified total momentum $P_\lambda=\lambda^{-1}P$ and new pre-factors for each term. 
If we factor out $\lambda^{\max\{\beta, \gamma\}}$, the coupling constant to the interaction becomes $g_\lambda=\lambda^{d/2-\alpha-\max\{\beta, \gamma\}}g$.

The large-momentum, respectively small-distance, behaviour of the model is related to the rescaled model with large $\lambda$. This has a small coupling constant if $d/2-\alpha-\max\{\beta, \gamma\}<0$, and we call this scaling subcritical ($d/2-\alpha-\max\{\beta, \gamma\} =0$, or $>0$, would be critical and supercritical, respectively).
The fact that the coupling becomes small for large $\lambda$ suggests that, concerning the UV behaviour, the interaction may be treated perturbatively. One expects such models to be superrenormalisable. This means, roughly speaking, that a renormalised model without UV cutoff can be defined after taking into account a finite number of divergences that are determined by perturbation theory. 
This is what we will show, under technical hypothesis that are sharp if $\beta=\gamma$ and in the precise sense of Theorem~\ref{thm:H} below.

\subsection{Main result}

We make the following assumptions on the dispersion relations and the interaction.

\begin{assumption}\label{ass:main} 
The functions $v:\R^d\to \R$, $\omega:\R^d\to \R_+$, $\Omega:\R^d\to \R_+$ are invariant under rotations.
 We assume that we have parameters, $\alpha<d/2$, $\gamma\in \{1,2\}$ satisfying
 \begin{equation*}
  \delta:=d-2\alpha-\gamma<\gamma,
 \end{equation*}
and constants $C>0$, $c_1>0$, $c_2\geq 0$ so that $\Omega\in C^\gamma(\R^d)$ and the inequalities
 \begin{equation*}
 \begin{aligned}
        |v(k)|&\leq C |k|^{-\alpha},   \\
        \omega(k) & \geq C (c_1 + k^2)^{\gamma/2}, \\
        \Omega(p)& \geq C (c_2 + p^2)^{\gamma/2}, \\
        |\partial^\nu\Omega(p)| &\leq C (c_2 + p^2)^{(\gamma-|\nu|)/2} , \qquad |\nu| \in \{1,  \gamma\},
 \end{aligned}
\end{equation*}
hold.
\end{assumption}
The first consequence of this hypothesis is that $v \omega^{-1}\in L^2(\R^d)$, since
\begin{equation}
 \int_{\R^d} \frac{|v(\xi)|^2}{\omega(\xi)^2} \ud \xi \leq C \int_{\R^d} \frac{\ud \xi}{|\xi|^{2\alpha} (c_1 + \xi^2)^{\gamma}},
\end{equation}
and $2\alpha + 2\gamma>d$. Note that the hypothesis on $\alpha, \gamma$ is sharp in this regard.
This implies that $a(v) \ud \Gamma(\omega)^{-1}$ is bounded (see Lemma~\ref{lem:a-omega}).
For $E_0\geq 0$ we set
\begin{equation}
 H_0:= \Omega(\ud \Gamma(k)-P)+\ud \Gamma(\omega)+E_0.
\end{equation}
Let $T$ be any operator such that $H_0+T$ is self-adjoint on $D(H_0)$ and invertible (for appropriate $E_0$). 
We define (this is well defined as under our hypothesis $H_0^{-1}a^*(v)$ is bounded by Lemma~\ref{lem:a-omega})
\begin{equation}
 G_T:=-(H_0+T)^{-1}a^*(v).
\end{equation}

\begin{thm}\label{thm:H}
Let $P\in \R^d$ and suppose $\Omega$, $\omega$, $v$ satisfy Assumption~\ref{ass:main}. 
 There exist a symmetric and $H_0$-bounded operator $(T,D(T))$ on $\cF$, a bounded symmetric operator $R\in \mathscr{L}(\cF)$, and a number $E_0$ such that the following hold.
 
 \begin{enumerate}[1)]
  \item\label{thm:sa} The operator 
\begin{equation*}
 H:=(1-G_T)^{*}(H_0+T)(1-G_T) + R-E_0
\end{equation*}
is self-adjoint on 
\begin{equation}
D(H)=\{ \Psi \in \cF: (1-G_T)\Psi \in D(H_0)\} 
\end{equation}
 and bounded from below.
\item\label{thm:conv} For $\Lambda\geq 0$, let $v_\Lambda(k)=v(k)1(|k|\leq \Lambda)$ and
 \begin{equation*}
  H_\Lambda:=\Omega(\ud \Gamma(k) - P)+\ud \Gamma(\omega)+a^*(v_\Lambda)+a(v_\Lambda).
 \end{equation*}
The operators $H_\Lambda$ are self-adjoint on $D(H_0)$ and for the numbers $E_\Lambda$ given in~\eqref{eq:E-def} we have
\begin{equation*}
\lim_{\Lambda\to \infty} \Big(H_\Lambda-E_\Lambda\Big)=H 
\end{equation*}
in the strong resolvent sense.
\item\label{thm:distr} 
For all $\Psi\in D(H)$ the equality
\begin{equation*}
H\Psi= \Omega(\ud \Gamma(k) - P)\Psi+\ud \Gamma(\omega)\Psi + a^*(v)\Psi + a(v)(1-G_T)\Psi + (T+R)\Psi 
 \end{equation*}
 holds in the dual of $D(H_0)$.
 \end{enumerate}
\end{thm}

This theorem is proved at the end of Section~\ref{sect:constr}.
The operator $T$ is constructed explicitly by an iterative procedure with number of iterations
\begin{equation*}
 n_* = \max\Big\{n\in\N: n(1-\delta/\gamma) \leq 1\Big\}=\Big\lfloor \frac{\gamma}{2\gamma+2\alpha-d} \Big\rfloor.
\end{equation*}
Note that our conditions on $\gamma$, $\alpha$ are exactly such that $n_*$ is finite.
The operator $R$ is given as a function of $T$ in Proposition~\ref{prop:R}.

The condition $(1-G_T)\Psi \in D(H_0)$ prescribes the behaviour of $\Psi \in D(H)$ at large momentum, respectively the singularities of its Fourier transform. Since $T$ preserves the boson number, it can be read as a the relation between
\begin{equation}
  \Psi\uppar{n}-G_T\Psi\uppar{n-1} \in D(H_0)
  \end{equation}
for $n\in \N$. Expanding the resolvent of $H_0+T$ in $G_T$, one finds the leading contribution 
\begin{equation}
  \sqrt{n} \Psi\uppar{n}(K, k_{n}) \sim \frac{v(k_n)}{\Omega(k_n)+\omega(k_n)} \Psi\uppar{n-1}(K)
\end{equation}
for $|k_n|\to \infty$. The presence of $T$ gives lower order corrections to this, see~\cite{La19, La20, La2021polaron}. Such conditions are known as interior-boundary conditions~\cite{TeTu20} and can be related to the general theory of self-adjoint extensions~\cite{BiLa21, posi20}.

The equality in point~\ref{thm:distr}) can be interpreted as follows. The operator $a(v)$ is defined on $D(H_0)$ and not 
\begin{equation}
D(H) \subset D(H_0)\oplus G_T \cF.
\end{equation}
If we extend it by setting 
\begin{equation}
 \overline{a}(v)G_T\Psi:=(T+R)\Psi,
\end{equation}
 where $T:\cF\to D(T)'\subset D(H_0)'$, then the equality reads
 \begin{equation}
  H\Psi= \Omega(\ud \Gamma(k) - P)\Psi+\ud \Gamma(\omega)\Psi + a^*(v)\Psi + \overline{a}(v)\Psi.
 \end{equation}
Note that here neither of the three terms needs to be an element of $\cF$. The individual terms are rather elements of $D(H_0)'$, with only the sum in $\cF$ due to cancellations enforced by the condition $(1-G_T)\Psi \in D(H_0)$.

Theorems on renormalisation with statements as in point~\ref{thm:conv}) of Theorem~\ref{thm:H} are abundant in the literature, see e.g.~\cite{nelson1964, Schrader1968, eckmann1970,gross1973, GuHiLo14, GrWu18, wuensch17, LaSch19, IBCrelat, IBCmassless, alvarez2021}. However, they are almost entirely constrained to cases which require only one renormalisation step, i.e.~with $n_*=1$. The recent results~\cite{GrWu18, wuensch17} are essentially sharp within this class, as the hypothesis amount roughly to $\delta<\tfrac\gamma2$ (see also Remark~\ref{rem:E_2}). A notable exception is the article~\cite{gross1973}, which treats the critical case $\beta=\gamma=\delta=1$. However, the renormalised operator $H$ is obtained by an abstract compactness argument, so very little is known about its properties (see also~\cite{deckert2014}).
The articles~\cite{La19, La20, La2021polaron} were the first to explicitly treat cases with $n_*=2$.
Models with $n_*=2$ that are slightly different from the polaron-type Hamiltonian were treated in~\cite{Schrader1968, IBCpaper} (see also Remark~\ref{rem:E_2}).

\begin{rem}\label{rem:beta neq gamma}
Our technical hypothesis are not optimal when $\omega$, $\Omega$ have different homogeneity, i.e., $\omega(k)\sim |k|^\beta$, $\Omega(k)\sim |k|^\gamma$, with $d-2\alpha-2 \max\{\beta, \gamma\}<0$ and $\beta\neq \gamma$. 
In this case, the iteration number should be
\begin{equation}
 n_* = \Big\lfloor \frac{\max\{\beta,\gamma\}}{2\max\{\beta,\gamma\}+2\alpha-d} \Big\rfloor.
\end{equation}
The generalisation of our results to the case $\beta > \gamma$ is straightforward and we omit it here in favor of a simpler notation since it seems to be less relevant to physics.
Concerning the case $\gamma>\beta$, the Nelson model (where $\beta=1$, $\gamma=2$ and $n_*=1$) was treated in~\cite{LaSch19}. 
 In general, this case brings with it the additional difficulty that, while $\Omega(\ud \Gamma(k)-P)$ can be used to remedy the lack of integrability of $v$, it does not control the number of bosons and $a(v)$. 
 Thus in all the relevant bounds there is a balance to be struck between gaining decay at $k\to \infty$ by using $\Omega$ and control of the boson number by $\ud \Gamma(\omega)$. 
 The number of iterations $n_*$ will be sufficient to render all expressions finite, for any given finite number of bosons.  However, the operator $R$ will not be bounded, but bounded by some power of the boson number $\cN$. 
 The requirement that this power be less than one will impose additional conditions on the range of $\delta$ depending on the difference $\gamma-\beta$.
 A similar statement to Theorem~\ref{thm:H} will still hold for appropriate values of $\delta$, but the proof will be more involved as it will require tracking a family of bounds and later selecting the appropriate ones, similarly to the considerations of~\cite[Sect.3.3]{LaSch19}.
 Alternatively, one could impose a cutoff on the boson number, and our method of proof should then work on the whole range of subcritical models (compare~\cite{La2021polaron}). 
 \end{rem}

\subsection{Notation}\label{sect:notation}

For normed spaces $X,Y$ we denote by $\mathscr{L}(X,Y)$ the normed space of bounded linear maps from $X$ to $Y$ and by $X'=\mathscr{L}(X,\C)$ the topological dual to $X$.
We denote by $\cF$ the symmetric Fock space over $L^2(\R^d)$,
\begin{equation}
 \cF:=\bigoplus_{n=0}^\infty L^2(\R^d)^{\otimes_{s}n} = \bigoplus_{n=0}^\infty \cF\uppar{n}.
\end{equation}
For a unitary $U\in \mathscr{L}(L^2(\R^d))$ we define the unitary $\Gamma(U)\in \mathscr{L}(\cF)$ by $\Gamma(U)\vert_{\cF\uppar{n}}=U^{\otimes n}$ and for a self-adjoint $A, D(A)$ on $L^2(\R^d)$ we define $\ud \Gamma(A)$ as the generator of $\Gamma(\ue^{-\ui t A})$. In particular, we denote by $\cN:=\ud\Gamma(1)$ the number operator on $\cF$, and by $\ud \Gamma(k)=(\ud \Gamma(k_1), \dots, \ud \Gamma(k_d))$ the vector-valued field momentum.

For real valued functions $f,g:X\to \R$ on a set $X$, we employ the notation
\begin{equation}
 f(x) \lesssim g(x)
\end{equation}
for the statement that there exists $C>0$ so that for all $x\in X$
\begin{equation}
 f(x) \leq C g(x).
\end{equation}
For symmetric linear operators the notation $A\lesssim B$ refers to the corresponding inequality of quadratic forms as usual. 

For a multi-index $\nu \in \N_0^k$ we denote by $|\nu|=\sum_{j=1}^k \nu_j$ the sum of its entries.
For a vector $X\in (\R^{d})^n$ or a multi-index $X\in \N_0^n$, we denote its components by the corresponding lower case letter, and by
\begin{equation}\label{eq:not-Xab}
 X_a^b:=(x_j)_{j=a}^b=(x_a, \dots, x_b)
\end{equation}
the element of $\R^{d(b-a+1)}$, respectively $\N_0^{b-a+1}$, given by the entries $x_j$ with $1\leq a\leq j\leq  b\leq n$. Similarly, we denote for an index set $M=(m_1,\dots , m_k)  \in \N^k$ with $m_i\neq m_j$ for $i\neq j$,
\begin{equation}\label{eq:not-X_J}
\begin{aligned}
 X_M&= (x_{m_1}, \dots , x_{m_k}), \\
 X_{M^c}&=(x_j)_{j\notin \{m_1, \dots, m_k)}.
\end{aligned} 
\end{equation}
Note that in $X_M$ the entries occur in the order given by the numbering of $m_1, \dots, m_k$, which can differ from the ordering of the $m_j$ by magnitude. For $k=n$, $X_M$ is a permutation of $X$. In $X_{M^c}$ entries occur in the same order as in $X$.

Moreover, we denote
 \begin{equation}\label{eq:not-omega}
  \omega(X_a^b):=\sum_{j=a}^b \omega(x_j).
 \end{equation}

\section{Construction of the Hamiltonian}\label{sect:constr}

In this section we give a detailed outline of the iterative construction procedure for the Hamiltonian, while postponing the technical work of proving the required bounds to Section~\ref{sect:bounds}. The general strategy of this construction should also be applicable to models with a different structure than the polaron models we consider, though of course the requirements for obtaining appropriate bounds will depend on the details. 

We begin by explaining how the auxiliary operator $T$ enters into the story. We will outline its construction below. Let $T$ be any operator such that $H_0+T$ is self-adjoint and invertible on $D(H_0)$, and set as above
\begin{equation}
 G_T:=- \big(a(v)(H_0+T)^{-1}\big)^* =-(H_0+T)^{-1}a^*(v).
\end{equation}
As $(H_0+T)G_T=-a^*(v)$ we have the following short calculation, which at this point is purely formal,
\begin{align}
 (1-G_T^*)(H_0 + T)(1-G_T)&= (H_0 + T + a(v))(1-G_T) \notag \\
 &= H_0 + a(v) + a^*(v)  + T - a(v)G_T.
\end{align}
Hence, if we were able to make sense of $ T - a(v)G_T$, say as a bounded operator, then we could define 
\begin{equation}
 H:=(1-G_T^*)(H_0 + T)(1-G_T) + a(v)G_T-T-E_0,
\end{equation}
and this would be self-adjoint on $(1-G_T)^{-1}D(H_0)$. The role of $T$ here is twofold. On the one hand, the domain of $H$ depends on $T$ via $G_T$, and on the other hand $T$ must be chosen so that $a(v)G_T-T$ has good properties. Since $a(v)G_T$ is not defined a priori, the latter means there are cancellations between the two terms.

In order to make this discussion well defined and explain the relation to renormalisation, we introduce a sharp UV-cutoff $\Lambda\geq 0$ and set 
\begin{equation}
 v_\Lambda(k):=v(k) 1(|k|\leq \Lambda).
\end{equation}
For any bounded operator $T_\Lambda$ we set 
\begin{equation}
 G_{T_\Lambda}:=-(H_0+T_\Lambda)^{-1}a^*(v_\Lambda),
\end{equation}
and then have the identity
\begin{align}
 H_\Lambda  
 &= \Omega(\ud\Gamma(k)-P) + \ud \Gamma(\omega)+ a(v_\Lambda)+a^*(v_\Lambda) \notag\\
 &= (1-G_{T_\Lambda}^*)(H_0 + T_\Lambda)(1-G_{T_\Lambda}) - T_\Lambda - a(v_\Lambda)(H_0+T_\Lambda)^{-1}a^*(v_\Lambda) - E_0.\label{eq:HLambda-G}
\end{align}
If $T_\Lambda$ has a limit $T$ for $\Lambda\to \infty$ so that, for example, $T$ is $H_0$-bounded with bound less than one, then the first term in the sum above defines a self-adjoint operator also for $\Lambda=\infty$. If we could choose the operators $T_\Lambda$ such that
\begin{equation}\label{eq:op-eq}
 a(v_\Lambda)G_{T_\Lambda}-T_\Lambda= - a(v_\Lambda)(H_0 + T_\Lambda)^{-1} a^*(v_\Lambda)- T_\Lambda = E_\Lambda
\end{equation}
for some numbers $E_\Lambda$, then we could give meaning to the operator $H_\Lambda - E_\Lambda$ for $\Lambda=\infty$.
It is not clear whether operators solving this equation can be found. However, we will show that under the assumption $\delta<\gamma$ we can iteratively construct approximate solutions, up to error terms that are eventually  well behaved for $\Lambda\to \infty$.
To find a formula for $T_\Lambda$, consider for a moment the interaction form factor $v_g=gv$. We then make an ansatz as a power series in $g^2$, i.e.,
\begin{equation}
 T_\Lambda = \sum_{j=1}^n g^{2j} T_{\Lambda, j}.
\end{equation}
That is, $T_{\Lambda, j}$ is homogeneous of degree $2j$ in $v$. The equation~\eqref{eq:op-eq} gives for $j=1$
\begin{equation}
 g^2 T_{\Lambda, 1} = - a(g v_\Lambda)H_0^{-1} a^*(gv_\Lambda) - g^2 E_{\Lambda, 1},
\end{equation}
with a natural choice for $E_{\Lambda, 1}$ given by
\begin{align}
 E_{\Lambda, 1}& = - \langle \varnothing,  a(v_\Lambda)H_0^{-1} a^*(v_\Lambda) \varnothing \rangle_{P=0}  \notag \\
 &= -\int\limits_{|k|\leq \Lambda} \frac{v(k)^2}{\Omega(k) + \omega(k)+E_0} \ud k \sim \Lambda^{\delta} \label{eq:self-E_1}.
\end{align}
Now $T_{\Lambda, j+1}$ should be the coefficient of $g^{2j+2}$ in
\begin{equation}
 - a(gv_\Lambda)\Big(H_0 + \sum_{\ell=1}^j  g^{2j} T_{\Lambda,\ell}\Big)^{-1} a^*(gv_\Lambda)- \sum_{\ell=1}^j g^{2\ell} T_{\Lambda,\ell} - g^{2j+1} E_{\Lambda, j+1}.
\end{equation}
To obtain an explicit formula, we use the following resolvent expansion that follows by induction from the resolvent formula, see~\cite[Lem.3.13]{bossmann2021}.
\begin{lem}\label{lem:res-gen}
 Let $H_0$ be closed and and set $H_n=H_0+\sum_{m=1}^n g^{2m} T_m$ for some $H_0$-bounded operators $T_m$, $m\in \{1, \dots, n\}$. For $-z\in \rho(H_0)\cap \rho(H_n)$ and any $L\in \N_0$ it holds
 \begin{equation*}
  (H_n+z)^{-1} = (H_0+z)^{-1} \sum_{\ell=0}^L  g^{2\ell} S_\ell  - g^{2L+2} \sum_{\ell=0}^L \sum_{j=\ell+1}^{n}(H_n+z)^{-1}T_j (H_0+z)^{-1} S_{L-\ell}
 \end{equation*}
with 
\begin{equation*}
S_0=1, \qquad S_{\ell}= \sum_{\nu=1}^\ell (-1)^\nu \sum_{1\leq j_1, \dots, j_\nu \leq n \atop j_1+\dots+j_\nu=\ell} \prod_{\mu=1}^\nu T_{j_\mu}(H_0+z)^{-1},\quad \ell\in \N.
\end{equation*}
\end{lem}
The expansion of Lemma~\ref{lem:res-gen} yields the recursive formula for $T_{\Lambda, j}$ (using the notation $J=(j_1, \dots, j_\nu)$)
\begin{equation}\label{eq:T_Lambda def}
 T_{\Lambda, j+1} + E_{\Lambda, j+1} = \sum_{\nu=1}^j (-1)^{\nu+1} 
 \sum_{J\in \{1, \dots, j\}^\nu \atop |J|=j}  a(v_\Lambda)H_0 ^{-1} \Big(\prod_{\ell=1}^\nu T_{\Lambda,j_\ell}H_0^{-1} \Big) a^*(v_\Lambda),
\end{equation}
where the numbers $E_{\Lambda, j}$ can be chosen to be the vacuum expectation values of $T_{\Lambda, j}$ with $P=0$. These numbers arise naturally in the formal perturbation series for the ground state energy at $P=0$. 
The coupling constant $g$ was only introduced here as a marker in the power counting argument, hence we again set $g\equiv1$ for the remainder of the article.

Roughly speaking, the terms in $T_{\Lambda, j}$ involve $j$ of each creation and annihilation operators and $2j-1$ resolvents of $H_0$.
If we were not missing one resolvent, we could expect these operators to be bounded, but this defect should become less and less important for large $j$. We thus expect more and more regular behaviour as $j$ increases. Moreover, the numbers $E_{\Lambda, j}$ should diverge more slowly, or not at all, for larger $j$. Simply adding up the homogeneities of the different factors, we expect that
\begin{equation}
 E_{\Lambda,n}\sim \Lambda^{2n(d/2-\alpha)-(2n-1)\gamma} = \Lambda^{\delta - (n-1)(\gamma-\delta)},
\end{equation}
when $v, \omega, \Omega$ are essentially homogeneous. 
Note that the exponent of $\Lambda$ is non-negative exactly if $n\leq n_*$, so there should be no more divergent contributions for $n>n_*$.

The ansatz~\eqref{eq:self-E_1} for $E_{\Lambda, 1}$ is well known, e.g. from the Nelson model, where it diverges logarithmically. When higher-order contributions $E_{\Lambda, j}$ are also divergent for $\Lambda\to \infty$ these higher orders must also be taken into account when constructing $H$. For example, in the Bogoliubov-Fröhlich model, $E_{\Lambda,1}\sim \Lambda$, $ E_{\Lambda,2}\sim\log(1+\Lambda)$ and since by the results of~\cite{La19, La20} the operators $H_\Lambda + E_{\Lambda, 1} + E_{\Lambda, 2}$ converge to $H$ as $\Lambda\to \infty$, $H_\Lambda  +E_{\Lambda, 1}$ cannot have a limit.

The important point is now that the definition of the operators $T_\Lambda$ can be extended to $\Lambda=\infty$. For this, it is necessary to understand the cancellations between the divergent sequences $E_{\Lambda,j}$ and the other terms in $T_{\Lambda, j}$. 
Consider $T_{\Lambda, 1}$. It involves one creation  operator on the right and an annihilation operator to the left.  To bring it into a sort of normal order, we use the pull-through formula
\begin{equation}\label{eq:pull}
 a_q (\Omega(\ud \Gamma(k)-P) + \ud \Gamma(\omega))^{-1} = (\Omega(\ud \Gamma(k)+q-P) + \ud \Gamma(\omega)+\omega(q))^{-1} a_q.
\end{equation}
With this, we obtain
\begin{align}
 T_{\Lambda,1} & = - a(v_\Lambda)(H_0 +E_0)^{-1} a^*(v_\Lambda) - E_{\Lambda, 1} \notag\\
 &= - \int \ud q \ud r\, a_r \frac{v_\Lambda(q) v_\Lambda(r)}{\Omega(\ud \Gamma(k)-P) + \ud \Gamma(\omega)+E_0} a^*_q- E_{\Lambda, 1} \notag \\
 &= -\int \ud q \ud r \, a_q^* \frac{v_\Lambda(q) v_\Lambda(r)}{\Omega(\ud \Gamma(k)+q+r-P) + \ud \Gamma(\omega)+\omega(q)+\omega(r)+E_0} a_r
 \label{eq:T_1,od}\\
 &\qquad- \int \ud q \ud r  \frac{v_\Lambda(q) v_\Lambda(r)\delta(q-r)}{\Omega(\ud \Gamma(k)+r-P) + \ud \Gamma(\omega)+ \omega(r)+E_0} - E_{\Lambda, 1}.\label{eq:T_1,d}
\end{align}
Note that~\eqref{eq:T_1,d} acts on the $n$-boson space as a multiplication operator, while~\eqref{eq:T_1,od}
acts essentially as an integral operator, and vanishes on the vacuum. Hence~\eqref{eq:T_1,od} will make sense even for $\Lambda=\infty$ on functions that decay sufficiently fast, while~\eqref{eq:T_1,d} will have a limit for $\Lambda\to \infty$ by choice of $  E_{\Lambda, 1}$ (note that this is the value of the integral with zero bosons and $P=0$).
For general $n=1, \dots, n_*$ we define
\begin{equation}
 T_{\Lambda, n} = \sum_{m=0}^n \theta_{\Lambda, n, m},
\end{equation}
where $\theta_{\Lambda, n, j}$ is the part of $T$ that, after normal ordering as above, contains exactly $j$ creation and annihilation operators and $E_{\Lambda,n}$ is absorbed into $\theta_{\Lambda,n, 0}$ (i.e., for $n=1$ we have $\theta_{\Lambda,1,0}=\eqref{eq:T_1,d}$ and $\theta_{\Lambda, 1,1}=\eqref{eq:T_1,od}$). The $\theta_{\Lambda,n,m}$ are given by a recursive formula, see~\eqref{eq:theta-nm-def} for $m>0$ and~\eqref{eq:theta-n0-def} for $m=0$.
For the operators $\theta_{\Lambda, n, 0}$ we may take the limit $\Lambda\to \infty$, while for 
$\theta_{\Lambda, n, j}$, $j\geq 1$, we can simply set $\Lambda= \infty$ and obtain an unbounded operator. This defines the operator $T_n:=T_{\infty, n}$.  

\begin{prop}\label{prop:T} 
 Let $n\leq n_*$. 
 \begin{enumerate}[a)]
  \item\label{prop:T-bound}
  For all $s<\tfrac12 n(1-\delta/\gamma)$ there exists a constant $C>0$ so that for all $\Lambda \in \R_+ \cup \{\infty\}$ and $\Psi \in D(H_0)$
 \begin{equation*}
  \|T_{\Lambda,n} \Psi \|_\cF \leq C\|(\ud \Gamma(\omega)+E_0)^{-s}H_0\Psi\|_\cF;
 \end{equation*}
 \item\label{prop:T-conv}
 For all $\Psi\in D(H_0)$
 \begin{equation*}
   \lim_{\Lambda\to \infty} \|T_{\Lambda,n} \Psi - T_n \Psi\|_\cF=0;
  \end{equation*}
  \item\label{prop:T-sym} $T_{\Lambda,n}$ defines a symmetric operator on $D(H_0)$ for $\Lambda \in \R_+ \cup \{\infty\}$.
  \end{enumerate}
 \end{prop}
This proposition is key to our proof of Theorem~\ref{thm:H} and requires the most work. It is proved in Section~\ref{sect:propT}. Note that by~\ref{prop:T-bound}), $T_{\Lambda, n}$ is bounded relative to $H_0$ uniformly in $\Lambda$, and for $E_0$ large enough the relative bound is less than one.  
 
We set 
\begin{equation}
T=\sum_{j=1}^{n_*} T_j, \qquad D(T):=D(H_0).
\end{equation}
We also define
\begin{equation}\label{eq:E-def}
 E_\Lambda:= \sum_{m=1}^{n_*} E_{\Lambda, m}
\end{equation}
with $E_{\Lambda, m}$ given explicitly in~\eqref{eq:E_n-def}.

From Proposition~\ref{prop:T} we know that $H_0+T_\Lambda$ is self-adjoint and positive for $E_0$ sufficiently large. We can thus make sense of $G_{T_\Lambda}$.

\begin{prop}\label{prop:G}
 For $E_0\geq 0$ sufficiently large and $\Lambda\in \R_+\cup\{\infty\}$, the operator
 \begin{equation*}
  G_{T_\Lambda}:=- (H_0+T_{\Lambda})^{-1}a^*(v)
 \end{equation*}
is bounded on $\cF$ uniformly in $\Lambda\in \R_+\cup \{\infty\}$. Moreover, $(1-G_{T_\Lambda})$ is invertible with uniformly bounded inverse and $G_{T_\Lambda}$ converges to $G_T$ strongly as $\Lambda\to \infty$.
\end{prop}
\begin{proof}
 First note that, by Proposition~\ref{prop:T}\ref{prop:T-bound}), we can choose $E_0$ sufficiently large so that for all $\Lambda \in \R_+\cup\{\infty\}$
\begin{equation}\label{eq:ResLambda-unif}
 \|(H_0+T_\Lambda)^{-1}\| + \|T_\Lambda(H_0+T_\Lambda)^{-1}\| \leq 1.
\end{equation}
Then, we use Lemma~\ref{lem:a-omega} to obtain for $\tfrac12(1+\delta/\gamma)<s<1$
 \begin{align}
  \|a(v)H_0^{-1}\| 
  &\leq \|v \omega^{-s}\|_{L^2}  \| \ud \Gamma(\omega)^{s}H_0^{-1}\| \notag\\
  & \leq \|v \omega^{-s}\|_{L^2} E_0^{s-1}.
 \end{align}
 
Hence 
\begin{equation}
 \| G_{T_\Lambda}^* \| =  \| a(v)H_0^{-1} T_{\Lambda}(H_0+T_\Lambda )^{-1}-a(v)H_0^{-1} \|\leq 2\|v \omega^{-s}\|_{L^2} E_0^{s-1}.
\end{equation}
The adjoint of this operator, that is $G_{T_\Lambda}$, is also bounded, with the same norm.
Moreover, for $E_0$ large enough, this norm is less than one and $1-G_{T_\Lambda}$ is invertible by Neumann series with uniformly bounded inverse.

Finally, 
\begin{equation}
 G_T^*-G_{T_\Lambda}^* = G_{T_\Lambda}^*(T_\Lambda-T)(H_0+T)^{-1}
\end{equation}
converges to zero strongly, because $(T_\Lambda-T)(H_0+T)^{-1}\to 0$ by Proposition~\ref{prop:T}\ref{prop:T-conv}) and $G_{T_\Lambda}^*$ is uniformly bounded.
\end{proof}

Using the resolvent expansion~\ref{lem:res-gen} we can calculate the remainder in the representation~\eqref{eq:HLambda-G} of $H_\Lambda-E_\Lambda$, i.e. with $E_\Lambda=\sum_{j=1}^{n_*} E_{\Lambda,j}$,
\begin{align}\label{eq:RLambda-def}
 R_\Lambda& := -a(v_\Lambda) (H_0+T_\Lambda)^{-1}a^*(v_\Lambda) - T_\Lambda -E_\Lambda \\
 &= \begin{aligned}[t]
     a(v_\Lambda) &(H_0+T_\Lambda)^{-1} \sum_{\ell=0}^{n_*-1} \sum_{j=\ell+1}^{n_*} T_{\Lambda,j} \\
     &\times \sum_{\nu=1}^{n_*-1-\ell}(-1)^{\nu}\sum_{J\subset\{1, \dots, n_*\}^\nu \atop |J|=n_*-1-\ell} \Big( \prod_{\mu=1}^\nu H_0^{-1} T_{\Lambda, j_\mu}\Big) H_0^{-1} a^*(v_\Lambda).
    \end{aligned}\notag
\end{align}
The individual factors in this expression are defined also for $\Lambda=\infty$ and using the same methods as in the proof of Proposition~\ref{prop:T} we can show that $R:=R_\infty$ is a bounded operator.

\begin{prop}\label{prop:R}
 The operator
 \begin{equation*}
  R=G_T^*  \sum_{\ell=0}^{n_*-1} \sum_{j=\ell+1}^{n_*} T_{j} \sum_{\nu=1}^{n_*-1-\ell}(-1)^{\nu}\sum_{J\subset\{1, \dots, n_*\}^\nu \atop |J|=n_*-1-\ell} \Big( \prod_{\mu=1}^\nu H_0^{-1} T_{j_\mu}\Big) G_0
 \end{equation*}
is bounded on $\cF$. The family $R_\Lambda$, $\Lambda\in \R_+ \cup\{\infty\}$ is uniformly bounded and converges to $R$ strongly as $\Lambda\to \infty$.
\end{prop}

This proposition is proved in Section~\ref{sect:propT}.
We now have all the necessary ingredients to prove Theorem~\ref{thm:H}.

\begin{proof}[Proof of Theorem~\ref{thm:H}]
 We first prove self-adjointness of 
 \begin{equation}
  H:=(1-G_T)^{*}(H_0+T)(1-G_T) + R-E_0
\end{equation}
on $D(H)$. By Proposition~\ref{prop:T}, $T=T_{n_*}$ is $H_0$-bounded with relative bound less than one. Thus for $E_0$ sufficiently large, $H_0+T$ is self-adjoint and invertible by the Kato-Rellich theorem.
 Choosing $E_0$ sufficiently large so that $(1-G_T)$ is invertible by Proposition~\ref{prop:G}, we then have that the operator
\begin{equation}
 (1-G_T)^{*}(H_0+T)(1-G_T) 
\end{equation}
is symmetric on $D(H)=(1-G_T)^{-1}D(H_0)$ and invertible. Hence it is self-adjoint, and as $R-E_0$ is a bounded symmetric operator this proves that $H$ is self-adjoint. 

We now prove that $H_\Lambda-E_\Lambda \to H$ in strong resolvent sense. First, note that $H_\Lambda$ is self-adjoint by the Kato-Rellich theorem, since
\begin{equation}
 \|a(v_\Lambda)\Psi + a^*(v_\Lambda)\Psi\|_\cF \leq \|v_\Lambda\|_{L^2} \| (\cN+1)^{1/2}\Psi \|_\cF \lesssim \| H_0^{1/2}\Psi \|_\cF .
\end{equation}

Now we write $H_\Lambda$ as in~\eqref{eq:HLambda-G} and use the definition of $R_\Lambda$ in~\eqref{eq:RLambda-def} to obtain
\begin{equation}
 H_\Lambda -E_\Lambda = (1-G_{T_\Lambda})^{*}(H_0+T_\Lambda)(1-G_{T_\Lambda}) + R_\Lambda - E_0.
\end{equation}
With this, we have (in the sense of quadratic forms)
\begin{align}
 &(H+\ui)^{-1} - (H_\Lambda -E_\Lambda + \ui)^{-1} \notag\\
&=(H_\Lambda -E_\Lambda + \ui)^{-1}(R_\Lambda-R)(H+\ui)^{-1} \label{eq:resdiff-R}\\
&\qquad+(H_\Lambda -E_\Lambda + \ui)^{-1}(1-G_T^*)(T_\Lambda-T)(1-G_T)(H+\ui)^{-1} \label{eq:resdiff-T}\\
&\qquad+(H_\Lambda -E_\Lambda + \ui)^{-1}(G_T^*-G_{T_\Lambda}^*)(H_0+T_\Lambda)(1-G_T)(H+\ui)^{-1}\label{eq:resdiff-G1} \\
&\qquad+(H_\Lambda -E_\Lambda + \ui)^{-1}(1-G_{T_\Lambda}^*)(H_0+T_\Lambda)(G_T^*-G_{T_\Lambda})(H+\ui)^{-1}\label{eq:resdiff-G2}.
\end{align}
The first line~\eqref{eq:resdiff-R} defines a bounded operator, and since $R_\Lambda-R\to 0$ strongly by Proposition~\ref{prop:R} this converges to zero strongly.
The second line~\eqref{eq:resdiff-T} also defines a bounded operator, because 
\begin{equation}
 (1-G_T)(H+\ui)^{-1} : \cF \to D(H_0)\subset D(T)
\end{equation}
is bounded, and $T_\Lambda-T$ converges strongly to zero on $D(H_0)$ by Proposition~\ref{prop:T}\ref{prop:T-conv}).

For~\eqref{eq:resdiff-G1} we have
\begin{align}
 \eqref{eq:resdiff-G1}
 &= (H_\Lambda -E_\Lambda + \ui)^{-1}(G_T^*-G_{T_\Lambda}^*)(H_0+T)(1-G_T)(H+\ui)^{-1} \\
 &\qquad + (H_\Lambda -E_\Lambda + \ui)^{-1}(G_T^*-G_{T_\Lambda}^*)(T-T_\Lambda)(1-G_T)(H+\ui)^{-1}.
\end{align}
The first term converges to zero because $G_T^*-G_{T_\Lambda}^*\to 0$ by Proposition~\ref{prop:G}, and the second one because $G_T^*-G_{T_\Lambda}^*$ is uniformly bounded by Proposition~\ref{prop:G} and $T-T_\Lambda \to 0$ by Proposition~\ref{prop:T}\ref{prop:T-conv}).

For the final term~\eqref{eq:resdiff-G2}, note that  
\begin{equation}
 (H_\Lambda -E_\Lambda + \ui)^{-1}(1-G_{T_\Lambda}^*)(H_0+T_\Lambda) = \Big((H_0+T_\Lambda)(1-G_{T_\Lambda})(H_\Lambda -E_\Lambda + \ui)^{-1}\Big)^*
\end{equation}
extends to a uniformly bounded operator on $\cF$, since 
\begin{align}
& (H_0+T_\Lambda)(1-G_{T_\Lambda})(H_\Lambda -E_\Lambda + \ui)^{-1} \notag\\
&= (1-G_{T_\Lambda}^*)^{-1} \Big (H_\Lambda-E_\Lambda - R_\Lambda \Big)(H_\Lambda -E_\Lambda + \ui)^{-1} \notag\\
&= (1-G_{T_\Lambda}^*)^{-1}\Big( 1- \ui - R_\Lambda \Big)(H_\Lambda -E_\Lambda + \ui)^{-1}
\end{align}
is uniformly bounded by Propositions~\ref{prop:G},~\ref{prop:R}.
Hence~\eqref{eq:resdiff-G2} converges to zero strongly for the same reason as~\eqref{eq:resdiff-G1}.
This proves strong resolvent convergence $H_\Lambda - E_\Lambda \to H$.
\end{proof}

\section{Construction of $T$}\label{sect:bounds}

As outlined above, the goal of this section is to prove Proposition~\ref{prop:T} by writing
\begin{equation}
 T_{n} = \sum_{j=0}^n \theta_{n, j},
\end{equation}
with an operator of multiplication $\theta_{n,0}$ and operators of integral type $\theta_{n,j}$, $j=1, \dots, n$. We then prove bounds on the kernels of $\theta_{n,j}$ and in an iteration step deduce from these bounds on the kernels of $\theta_{n+1,j}$.
To present these arguments in a clear way, we first introduce the relevant spaces of operators. 

\subsection{Integral operators on $\cF$ and their products}

 We denote $\cK_0$ the vector space of normal operators on $\cF$ acting as
 \begin{equation}\label{eq:kernel0}
 \kappa\Big(\ud \Gamma(k)-P, \ud \Gamma(\omega)+E_0\Big)
 \end{equation}
for a locally bounded function $\kappa:\R^d\times \R_+\to \C$.
Note that in particular $H_0\in \cK_0$ with the function given by
\begin{equation}
 H_0(p,E)=\Omega(p)+E,
\end{equation}
and also 
\begin{equation}\label{eq:H_0inv function}
 H_0^{-1}(p,E)= \frac{1}{\Omega(p)+E}.
\end{equation}
It is also not difficult to see that (see Lemma~\ref{lem:theta_10} below)
\begin{align}
 \theta_{1,0}:&=\lim_{\Lambda \to \infty} \theta_{\Lambda,1,0} \notag\\
 &=  - \lim_{\Lambda \to \infty} \bigg(\int   \frac{|v_\Lambda(\xi)|^2 \ud \xi}{\Omega(\ud \Gamma(k)+\xi-P) + \ud \Gamma(\omega)+E_0+ \omega(\xi)} + E_{\Lambda, 1}\bigg)
 \label{eq:theta-10}
\end{align}
converges to a locally bounded function of $\ud \Gamma(k)-P$, $\ud\Gamma(\omega)+E_0$ and thus defines an element of $\cK_0$.

For $j=1, \dots, n$, the operators $\theta_{n,j}$ will act as
\begin{equation}\label{eq:kernel-a*a}
 \kappa \Psi = \int\limits_{\R^{dm}\times \R^{d\ell}} \Big(\prod_{i=1}^m a^*_{q_i}\Big)\kappa\Big(Q,R, \ud\Gamma(k)-P ,\ud \Gamma(\omega)+E_0\Big)\Big(\prod_{i=1}^\ell a_{r_i}\Big)\Psi\, \ud Q \ud R 
\end{equation}
for some $\kappa \in L^2_\mathrm{loc}(\R^{md} \times \R^{\ell d}\times \R^d \times \R)$ with $m=\ell=j$.
Note that the action of $a(v)$, $a^*(v)$ can also be represented in the form~\eqref{eq:kernel-a*a}, with $(m,\ell)=(0,1)$ and $(m,\ell)=(1,0)$, respectively.
In the following, we will use the same notation for an operator of this type and its (operator-valued) integral kernel.

We will assume bounds on such kernels that ensure that $\kappa$ defines an operator from some dense domain $D\subset \cF$ to $\cF$.
Denote for $n\in \N$ and $\lambda \in \R$ 
\begin{align}\label{eq:rho-def}
\begin{aligned}
 \rho_{n,\lambda}(Q,E) & 
  = \left(\prod_{j=1}^{n-1}\frac{|v(q_j)|}{E+\omega(Q_j^n)} \right)\frac{|v(q_n)|}{(E+\omega(q_n))^{(1+\lambda)/2}} \\
  \tilde \rho_{n, \lambda}(R,E) &= \frac{|v(r_1)|}{(E+\omega(r_1))^{(1+\lambda)/2}}\left(\prod_{j=2}^{n}\frac{|v(r_j)|}{E+\omega(R_1^j)} \right) .
\end{aligned}
\end{align}
with the notation defined in~\eqref{eq:not-Xab},~\eqref{eq:not-omega}.

\begin{definition}\label{def:K_n}
Let $n\in \N_0$.
\begin{itemize}
 \item  For $n>0$ the space $\cK_{n}=\cK_{n, 0}$ is the space of operators acting as~\eqref{eq:kernel-a*a} with $m=\ell=n$ whose kernels satisfy
 \begin{equation*}
  | \kappa(Q,R,p,E)| \lesssim \min_{s\in [-1, 1]} \rho_{n, s}(Q,E) \tilde \rho_{n, -s}(R,E).
 \end{equation*}
 \item For $n>0$, $\lambda> 0$ the space $\cK_{n,\lambda}$ is the subspace of $\cK_n$ such that $\kappa\in\cK_{n,\lambda}$ satisfies for all $0\leq \sigma \leq 1$ with $\sigma< \lambda$ and $E\geq 1$
\begin{equation*}
 | \kappa(Q,R,p,E)| \lesssim E^{-(\lambda-1)_+} \min_{s\in [\sigma -1, 1-\sigma]} \rho_{n, \sigma + s}(Q,E) \tilde \rho_{n, \sigma-s}(R,E).
\end{equation*}
\item For $n=0$ and $\lambda\geq 0$ we denote by $\cK_{0, \lambda}$ the subspace of $\cK_0$ such that for $\kappa \in \cK_{0, \lambda}$ we have for all $0\leq \sigma<\lambda$ and $E\geq 1$
\begin{equation*}
|\kappa(p,E)| \lesssim H_0(p, E) E^{-\sigma}.
\end{equation*}
\end{itemize}
\end{definition}

By Lemma~\ref{lem:K_n-op} an element of $\cK_n$, $n\geq 1$, defines an operator from $D(H_0)$ to $\cF$.
Note that $\cK_{n, \lambda} \subset \cK_{n, \lambda'}$ for $\lambda' \leq \lambda$.

Note from~\eqref{eq:T_1,od} that $\theta_{1,1}$ has kernel
\begin{equation}\label{eq:theta-11}
 \theta_{\Lambda, 1,1}(q,r,p,E)= -\frac{v_\Lambda(q)v_\Lambda(r)}{\Omega(p+q+r)+E + \omega(q) + \omega(r)},
\end{equation}
so it is clearly an element of $\cK_{1,0}$ for $\Lambda \in \R_+\cup \{\infty\}$.

In order to construct the kernels $\theta_{n,j}$ for $n>1$ we will need to take products of such operators.
We will see that these can again be expressed as linear combinations of elements in $\cK_j$ for different $j$'s. Moreover, one gains some decay of the kernels in this process.

 If we take the product $\kappa \kappa'$ of $\kappa\in \cK_n$, $\kappa'\in \cK_m$, we will need to commute all the creation operators in $\kappa'$ to the left and the annihilation operators in $\kappa$ to the right  in order to put the product into the form~\eqref{eq:kernel-a*a}. Since the commutator of $a_r$ with  $a^*_{q'}$ is $\delta(q'-r)$, this leads to ``contractions'' between a variable $r$ of $\kappa$ and a variable $q'$ of $\kappa'$. That is, some of the integrals become part of the definition of a new integral kernel and do not involve variables of $\Psi$ any more. Moreover, when we commute creation or annihilation operators with functions of $\ud \Gamma(k)$ and $\ud \Gamma(\omega)$, we must take into account the pull-through formulas
\begin{equation}\label{eq:pull-gen}
 \begin{aligned}
  a_r f(\ud \Gamma(k)) &=  f(\ud \Gamma(k)+r) a_r \\
  a_r f(\ud \Gamma(\omega)) &= f(\ud \Gamma(\omega)+\omega(r))a_r,
 \end{aligned}
\end{equation}
as well as their adjoint relations (these hold by inspection in the sense that when acting on $\Psi \in D(\cN^{1/2}f(\cdot))$ both sides give the same element of $L^2(\R^d, \ud r)\otimes \cF$).

With this, we can write
\begin{equation}
 \kappa  \kappa' = \sum_{\ell=0}^{\min\{n,m\}} \kappa \star_\ell \kappa',
\end{equation}
where $\star_\ell$ denotes the part of the product with exactly $\ell$ contractions (i.e., $\ell$ commutators $[a_{r_i}, a^*_{q_j'}]$). When there are exactly $\ell$ contractions, these involve $\ell$ of the components of $R$, $r_{i_1}, \dots, r_{i_\ell}$, with $i_1<i_2< \dots <i_\ell$. Each of these is then paired with one of the components of $Q$, i.e. $q_{j_1}, \dots, q_{j_\ell}$, where the $j_\mu$ are pairwise different but otherwise arbitrary. Summing over all possibilities, the kernel of $\kappa \star_\ell \kappa'$ takes the form  
\begin{align}
 &\kappa \star_\ell \kappa'(Q, R,p, E) \label{eq:kappa-star_ell} \\
 &=\sum_{1\leq i_1< \dots <i_\ell\leq n \atop 1 \leq j_1\neq \dots \neq j_\ell \leq m} \int\limits_{\R^{\ell d}} 
  \kappa\bigg(Q_1^n, S, p+ \sum_{\nu=n+1}^{n+m-\ell} q_\nu, E + \sum_{\nu=n+1}^{n+m-\ell} \omega(q_\nu)\bigg)
 \bigg\vert_{S_I=\Xi \atop S_{I^c}=R_1^{n-\ell}} \notag\\
 &\qquad \times \kappa'\bigg( U, R_{n-\ell+1}^{n+m-\ell}, p+\sum_{\mu=1}^{n-\ell} r_\mu, E + \sum_{\mu=1}^{n-\ell} \omega(r_\mu)\bigg)
 \bigg\vert_{U_J=\Xi \atop U_{J^c}=Q_{n+1}^{n+m-\ell}} \ud \Xi  \notag,
\end{align}
 where we used the notation~\eqref{eq:not-X_J}. We will also use the notation $\kappa \star_\ell \kappa'$ for products with $\ell$ contractions of operators acting as in~\eqref{eq:kernel-a*a} for which $Q$, $R$ have different dimensions, notably $a(v)$, $a^*(v)$ . The expression for the kernel $\kappa \star_\ell \kappa'$ in this case can be easily obtained from~\eqref{eq:kappa-star_ell} by appropriately adjusting the arguments (see Lemma~\ref{lem:a-star} for the precise formulas).
 
%
 
 Using the bound provided by the definition of $\cK_n$, it is not difficult to see that the integral over $\Xi$ in~\eqref{eq:kappa-star_ell} converges, since we can always choose the parameter $s$ in that bound in such a way that there is a resolvent for each contracted variable (i.e., take $s=-1$ for $\kappa$ and $s'=1$ for $\kappa'$). However, in general, $\kappa\star_\ell \kappa'$ is not an element of $\cK_{n+m-\ell}$. For us the following will be more important.
 
\begin{thm}\label{thm:K-star-prod}
Let $E_0\geq 1$ and let $n,m\in \N_0$, $\lambda, \lambda'\geq 0$, and $\kappa\in \cK_{n,\lambda}$, $\kappa'\in\cK_{m,\lambda'}$. Then for $\ell \leq \min\{n,m\}$
 \begin{equation*}
  \kappa H_0^{-1}\star_\ell \kappa' \in \cK_{\mu, \sigma},
  \end{equation*}
  where
  \begin{align*}
   \mu&=m+n-\ell, \\
   \sigma&=\lambda + \lambda' + \ell(1-\delta/\gamma).
  \end{align*}

\end{thm}
 \begin{proof}
 The basic idea is that every contraction comes with an integral
 \begin{equation}
  \int  \frac{|v(\xi)|^2}{(E+\omega(\xi))^2} \ud \xi \lesssim E^{-(1-\delta/\gamma)}. 
 \end{equation}
  The details of the bounds on the integrals are given in Section~\ref{sect:kappa-prod} of the appendix.
  The case $\ell=0$ follows from Lemma~\ref{lem:star ell=0} (respectively the fact that the kernel is just the product of kernels in case $m=n=0$). The case $0<\ell < \max\{n, m\}$ is obtained from Lemma~\ref{lem:star ell<n} (the case excluded there does not arise since $1+\delta/\gamma>1$). The remaining case $\ell=m=n$ is Lemma~\ref{lem:star ell=n} (the loss of decay in this lemma is the main reason why  we restrict to $\sigma <\lambda$ in the definition of $\cK_{n, \lambda}$). 
 \end{proof}

If we set $\cK:=\bigoplus_{n=0}^\infty \cK_n$, then 
\begin{equation}
 (\kappa, \kappa')\mapsto \kappa H_0^{-1} \star_\ell \kappa'
\end{equation}
defines a binary operation on $\cK$. We will view this as a product, even though it is not, in general, associative. Associativity can fail for $\ell\neq 0$ since the expression
\begin{equation}
 \kappa\star_1 (\kappa'\star_1 \kappa'')
\end{equation}
has contributions with one or zero contractions between $\kappa$ and $\kappa'$, whereas in
\begin{equation}
 (\kappa\star_1 \kappa')\star_1 \kappa''
\end{equation}
 there is exactly one contraction between $\kappa$ and $\kappa'$.
 However, in the special case $m=n=\ell$, the operation
 \begin{equation}
 \cK_n\times \cK_n \to \cK_n, \qquad (\kappa, \kappa')\mapsto \kappa \star_n H_0^{-1}\kappa'
\end{equation}
 is associative. The reason is that in $\kappa H_0^{-1} \star_n \kappa'$ all the available variables are contracted, and thus  all the additional contractions in
 \begin{equation}
  (\kappa H_0^{-1} \star_n \kappa')H_0^{-1} \star_n \kappa''
 \end{equation}
occur between $\kappa'$ and $\kappa''$.
 
Similar classes of integral operators and products with contractions occur throughout the literature on quantum field theory. Let us emphasize, however, that the classes $\cK_n$ and the product $\star_\ell$ differ from the common ones (see e.g.~\cite{glimm1968b}) in the fact that the kernel additionally depends on $\ud \Gamma(k), \ud \Gamma(\omega)$ and thus an element of $\cK_n$ acts non-trivially on all bosons, not just $n$ of them. This is crucial for our argument, as it allows us to shift decay from one variable to another as needed.

 \subsection{The algorithm for $T_n$} \label{sect:algorithm}

 We now describe the algorithm for the construction of $T_n$. Assume that we have given, for $\Lambda \in \R_+ \cup\{\infty\}$ and $k\leq n$
 \begin{equation}
  T_{\Lambda, k} = \sum_{j=0}^k \theta_{\Lambda, k, j}
 \end{equation}
with $\theta_{\Lambda, k, j}\in \cK_j$.
We now want to write explicit formulas for $\theta_{n+1, m}$, $m=0, \dots, n+1$. Our task will then be to derive appropriate bounds on these kernels.

In view of~\eqref{eq:T_Lambda def}, the operator $\theta_{\Lambda, n+1, m}$ is the sum over $\nu, J=(j_1, \dots, j_\nu)\in \{1,\dots, n\}^\nu, |J|=n$ of those terms in the expression 
\begin{align}
a(v_\Lambda) H_0^{-1} \left(\prod_{\mu=1}^\nu  T_{\Lambda,j_\mu} H_0^{-1}\right)  a^*(v_\Lambda)
=a(v_\Lambda) H_0^{-1}  \left(\prod_{\mu=1}^\nu \sum_{i_\mu=0}^{j_\mu}\theta_{\Lambda,j_\mu, i_\mu}H_0^{-1} \right)  a^*(v_\Lambda)
\end{align}
which have exactly $m$ creation (and annihilation) operators after putting these in normal order.
For a given sequence $i_1, \dots , i_\nu$,  there are a total of
$ 1+ \sum_{\mu=1}^\nu i_\mu$
creation operators, so there must be exactly $1+\sum_{\mu=1}^\nu i_\mu-m$ contractions.
Fixing $I=(i_1, \dots, i_\nu)$,  the corresponding term can be expressed as the sum over $L=(\ell_0,\ell_1, \dots, \ell_{\nu})$ of
 \begin{align}\label{eq:tau-def}
  \bigg( a(v) H_0^{-1}\star_{\ell_0}  \underbrace{\Big( \big(\theta_{\Lambda,j_1, i_1} H_0^{-1}\star_{\ell_1} \theta_{\Lambda,j_2, i_2}\big)  H_0^{-1} \star_{\ell_2} \cdots \Big)}_{=:\tau_{\Lambda,I,L_1^{\nu-1}}} \bigg) H_0^{-1} \star_{\ell_{\nu}}  a^*(v),
 \end{align}
 such that
 \begin{equation}\label{eq:L-sum}
  \sum_{\mu=0}^{\nu} \ell_\nu=1+\sum_{\mu=1}^\nu i_\mu-m,
 \end{equation}
with the constraints (imposed by the fact that we must have $\ell \leq \min \{n, m\}$ in~\eqref{eq:kappa-star_ell}),
\begin{equation}\label{eq:L-constr}
\begin{aligned}
&0\leq \ell_\mu \leq \min\Big\{\sum_{\iota=1}^{\mu} i_\iota -\sum_{\iota=1}^{\mu-1}\ell_\iota ,i_{\mu+1}\Big\} \text{ for } 0<\mu<\nu \\
&0 \leq \ell_0 \leq \min\Big\{1,  \sum_{\iota=1}^{\nu} i_\iota -\sum_{\iota=1}^{\nu-1} \ell_\iota \Big\} \\
&0 \leq \ell_\nu \leq \min\Big\{1,  1 +\sum_{\iota=1}^{\nu} i_\iota -\sum_{\iota=0}^{\nu-1} \ell_\iota \Big\}.
\end{aligned}
\end{equation}
Here we can take $\Lambda \in \R_+ \cup \{\infty\}$, since by Theorem~\ref{thm:K-star-prod}
and our knowledge about $T_k$, $k\leq n$ we have $\tau_{I, L_1^{\nu-1}} \in \cK_{m+\ell_0+\ell_\nu}$.
To be completely precise, we can define $\tau_{I,L}$ recursively by 
setting $\tau_{\Lambda,i_1,\varnothing}=\theta_{\Lambda,j_1, i_1}$ and
 \begin{equation}\label{eq:kernel-recursion}
  \tau_{\Lambda,(i_1, \dots, i_{\mu+1}), (\ell_1, \dots \ell_\mu)} = \tau_{\Lambda,(i_1, \dots, i_{\mu}), (\ell_1, \dots \ell_{\mu-1})} H_0^{-1} \star_{\ell_\mu} \theta_{j_{\mu+1}, i_{\mu+1}}.
 \end{equation}
For $J\in \{1, \dots, n\}^\nu$, $I\in  \N_0^\nu$ and $L$ as above, we then set 
\begin{equation}\label{eq:vartheta-def}
 \vartheta_{\Lambda, J, I, L} =\Big( a(v_\Lambda) H_0^{-1}\star_{\ell_0} \tau_{\Lambda, I,L_1^{\nu-1}} \Big) H_0^{-1} \star_{\ell_{\nu}}  a^*(v_\Lambda),
\end{equation}
and for $m=1, \dots, n+1$
\begin{align}\label{eq:theta-nm-def}
 \theta_{\Lambda, n+1, m} 
 := \sum_{\nu=1}^n  (-1)^{\nu+1}
 \sum_{J \in \{1, \dots, n\}^\nu \atop |J|=n}
 \sum_{I \in \N_0^\nu \atop  i_\mu \leq j_\mu}
 \sum_{L\in \N_0^{\nu+1} \atop \text{ with \eqref{eq:L-sum},~\eqref{eq:L-constr}}}
 \vartheta_{\Lambda, J, I, L}.
 \end{align}

 We will see in Theorem~\ref{thm:theta} below that $\theta_{\Lambda, n+1, m}$ is well defined, including for $\Lambda=\infty$ and $\theta_{n+1, m}\in \cK_{m}$. The basic reason for this is that there always remain uncontracted variables, and using the bound from Definition~\ref{def:K_n} we can then pair every contracted variable with two instances of $H_0^{-1}$ that contain this variable. 
 Note that the sums simplify considerably in some cases, e.g. for $m=n+1$ the only possibility for given $\nu$, $J$ is $i_\mu=j_\mu$ and $L=(0, \dots, 0)$.
Note also that 
\begin{equation}
 \langle \varnothing, \theta_{\Lambda, n+1, m} \varnothing \rangle = 0, \qquad m\neq 0.
\end{equation}

For $m=0$, the expression~\eqref{eq:theta-nm-def} is not well defined for $\Lambda=\infty$. In order to remedy this, we need to subtract its vacuum expectation value, the numbers $E_{\Lambda,n}$, and take the limit $\Lambda \to \infty$. For $m=0$, the constraint~\eqref{eq:L-sum} for $m=\nu$ becomes
\begin{equation}
0= 1 + \sum_{\mu=1}^\nu i_\mu - \sum_{\mu=0}^\nu \ell_\mu =  \underbrace{\sum_{\mu=1}^\nu i_\mu - \sum_{\mu=1}^{\nu-1} \ell_\mu - \ell_0}_{\stackrel{\eqref{eq:L-constr}}{\geq} 0}+ 1 - \ell_\nu.
\end{equation}
We deduce that $\ell_\nu=1$ and 
\begin{equation}
 \ell_0 = \sum_{\mu=1}^\nu i_\mu - \sum_{\mu=1}^{\nu-1} \ell_\mu = \underbrace{ \sum_{\mu=1}^{\nu-1} i_\mu - \sum_{\mu=1}^{\nu-2} \ell_\mu  -\ell_{\nu-1}}_{\stackrel{\eqref{eq:L-constr}}{\geq} 0} + i_\nu.
\end{equation}
From this and $\ell_\nu \leq 1$ we deduce that $i_\nu\leq 1$ and thus also $\ell_{\nu-1}\leq 1$. Recursively we thus obtain that $i_\mu\in \{0,1\}$, $\ell_\mu\in\{0,1\}$.
Moreover, $\ell_0=1$ exactly if at least one of $i_\nu$ and $\sum_{\mu=1}^{\nu-1} i_\mu - \sum_{\mu=1}^{\nu-2} \ell_\mu$ equal one. If both are equal to one, then we also have $\ell_{\nu-1}=1$. 
Then we see that $\ell_\mu=1$ if and only if $i_1=1=i_2$, and then 
\begin{equation}
 \sum_{\mu=1}^2 i_\mu - \ell_1=1 = \|I_1^2\|_\infty.
\end{equation}
Continuing with this reasoning, we see that for given $I\in \{0,1\}^\nu$ this leaves only one possible choice of $L$, which amounts to taking $\ell_\mu=1$ at every position where this makes sense, and $\ell_\mu=0$ otherwise.
The constraints~\eqref{eq:L-sum},~\eqref{eq:L-constr} become for $m=0$
\begin{equation}\label{eq:L-constr-0}
 \begin{aligned}
 \ell_0 &=\|I\|_\infty \\
 \ell_\mu&= \min\{ \|I_1^{\mu}\|_\infty, i_{\mu+1}\}\\
  \ell_\nu&=1.
 \end{aligned}
\end{equation}

We thus set for $\Lambda\in \R_+$ and $L=(\ell_0, \dots, \ell_\nu)$ given by~\eqref{eq:L-constr-0}
\begin{align}
 E_{\Lambda, n+1} 
 &:= \sum_{\nu=1}^n (-1)^{\nu+1} 
 \sum_{J \in \{1, \dots, n\}^\nu \atop |J|=n}
 \sum_{I \in \{0,1\}^\nu \atop  i_\mu \leq j_\mu}
  %
  \underbrace{\langle \varnothing, \vartheta_{\Lambda, J, I, L} \varnothing\rangle\vert_{P=0}}_{=:E_{\Lambda, J,I,L}} , \label{eq:E_n-def}
 \end{align}
 and
 \begin{align}\label{eq:theta-n0-def}
 \theta_{\Lambda,n+1,0} 
 := \sum_{\nu=1}^n (-1)^\nu
 \sum_{J \in \{1, \dots, n\}^\nu \atop |J|=n}
 \sum_{I \in \{0,1\}^\nu \atop  i_\mu \leq j_\mu}
  (\vartheta_{\Lambda, J, I, L} -  E_{\Lambda, J,I,L}), 
  %
 \end{align}
 where $\vartheta_{\Lambda, J, I, L}$ is given by~\eqref{eq:vartheta-def}.
 We will show below that the limit 
 \begin{equation}
  \lim_{\Lambda \to \infty} \theta_{\Lambda,n+1,0}
 \end{equation}
exists, and defines an element of $\cK_0$ satisfying appropriate bounds.

\begin{rem}\label{rem:E_2}
As an example note that $E_{\Lambda,2}$ hast two contributions, coming from $i=i_1\in \{0,1\}$, namely 
\begin{align}
 E_{\Lambda,2}& = 
  \langle \varnothing , a(v_\Lambda) H_0^{-1}\theta_{\Lambda,1,0}H_0^{-1} \star_1 a^*(v_\Lambda) \varnothing\rangle \notag\\
 &\qquad  + \langle \varnothing , a(v_\Lambda) H_0^{-1}\star_1\theta_{\Lambda,1,1} H_0^{-1}\star_1 a^*(v_\Lambda) \varnothing\rangle.
\end{align}
Observe that $\theta_{\Lambda,1,0}$ involves $E_{\Lambda,1}$ via $\theta_{\Lambda, 1,0}$ so $E_{\Lambda, 2}$ encodes the ``nested'' divergence after $E_{\Lambda,1}$ has been subtracted from $H_\Lambda$.

The second term here is negative, since $\theta_{\Lambda,1,1}$ has a negative integral kernel (see~\eqref{eq:theta-11}), and the first term is positive for large $\Lambda$ (if it diverges), because $\theta_{1,0}$ (cf.~\eqref{eq:theta-10}) is bounded from below but not necessarily from above. It is thus possible to have cancellations between the two for some special cases. This is exactly what happens for the  models considered in~\cite{IBCpaper, Schrader1968}, which correspond to $\gamma=2$, $\delta=1$. The scaling behaviour would suggest that $E_{\Lambda,2}$ diverges logarithmically as in~\cite{La19, La20}. However, the special choice of the masses (``infinite'' mass of the particles in~\cite{IBCpaper}, mass preservation upon boson creation in~\cite{Schrader1968}) arranges for the divergences to cancel each other and $E_{\Lambda,2}$ has a finite limit for $\Lambda\to \infty$.
 \end{rem}
 
 \begin{lem}\label{lem:theta_10}
 Assume the hypothesis~\ref{ass:main}. The expression~\eqref{eq:theta-10} is well defined and
 \begin{equation*}
\theta_{1,0} \in \cK_{0, (1-\delta/\gamma)}.
 \end{equation*}
Moreover, $(p, E)\mapsto \theta_{1,0}(p,E)$ is continuously differentiable in $E$ and $\gamma$-times differentiable in $p$ with
\begin{align*}
 \partial_E \theta_{1,0} &\in \cK_{0, 1+(1-\delta/\gamma)} \, \\
 \partial_p^\nu \theta_{1,0} &\in \cK_{0, |\nu|/\gamma+(1-\delta/\gamma)}.
\end{align*}
 \end{lem}
\begin{proof}
We give the proof for for $\delta\geq1$ (which only occurs for $\gamma=2$). When $\delta<1$ the development to second order in~\eqref{eq:theta_10 Hess} is not necessary, which simplifies the argument.
Note that $E_{\Lambda, 1}$ corresponds to the evaluation of the first term in~\eqref{eq:theta-10} at $p=0$, $E=E_0$.
By Taylor expansion we thus have
 \begin{align}
  &-\theta_{1,0}(p,E) \notag\\
  &=  \lim_{\Lambda \to \infty} \int\limits_{|\xi|\leq \Lambda} |v(\xi)|^2  \int_{E_0}^E \left(\partial_t \frac{1}{\Omega(\xi) + \omega(\xi) +t}\right)\ud t\ud \xi \label{eq:theta_10 E}\\
  &\qquad + \lim_{\Lambda \to \infty} \int\limits_{|\xi|\leq \Lambda} |v(\xi)|^2\bigg\langle p, \nabla_q\frac{1}{\Omega(\xi+q) + \omega(\xi)+E}\Big\vert_{q=0}\ud \xi\bigg\rangle\label{eq:theta_10 rot} \\
  &\qquad + \lim_{\Lambda \to \infty} \int\limits_{|\xi|\leq \Lambda} |v(\xi)|^2\int_0^1(1-t)\bigg\langle p,  \nabla_q^2 \frac{1}{\Omega(\xi+q) + \omega(\xi)+E}\Big\vert_{q=tp}  p\bigg\rangle\ud t\ud \xi \label{eq:theta_10 Hess}.
 \end{align}
We have (since $\delta>0$)
\begin{align}
 &\bigg|\lim_{\Lambda \to \infty} \int_{|\xi|\leq \Lambda} |v(\xi)|^2  \int_{E_0}^E \left(\partial_t \frac{1}{|\Omega(\xi) + \omega(\xi) +t }\right)\ud t\ud \xi \bigg| \notag \\
 &\leq \int_0^E \int_{\R^d} \frac{|v(\xi)|^2}{(\omega(\xi)+t)^2}\ud \xi \ud t \stackrel{\eqref{eq:integral-E}}{\lesssim}  \int_0^E \frac{1}{t^{1-\delta/\gamma}} \lesssim  E^{1-(1-\delta/\gamma)}.
\end{align}
For the derivatives in $p$, first note that 
\begin{equation}
  \nabla_q\frac{1}{\Omega(\xi+q)+ \omega(\xi)+E }\Big\vert_{q=0} =- \frac{ (\nabla \Omega)(\xi) }{(\Omega(\xi)+\omega(\xi)+E )^2}.
\end{equation}
Since $\Omega$ is rotation invariant, $\nabla \Omega(\xi)$ is proportional to $\xi/|\xi|$, and by rotation invariance of $v, \omega$ the integral~\eqref{eq:theta_10 rot} equals zero (note however that the integrand in~\eqref{eq:theta_10 rot} may not be absolutely integrable, so it is necessary to interpret the limit as an improper integral). 
For the remaining term~\eqref{eq:theta_10 Hess} note that the Hessian satisfies the bound
\begin{align}
\left| \nabla_p^2 \frac{1}{\Omega(p+\xi) + \omega(\xi)+E } \right| \lesssim \frac{1}{(\Omega(p+\xi) + \omega(\xi)+E )^2}.
\end{align}
Hence the integrand in~\eqref{eq:theta_10 Hess} integrable (in $(\xi,t)$) and we have 
\begin{align}
 \eqref{eq:theta_10 Hess}&\lesssim p^2   \int_0^1 (1-t) \int_{\R^d} \frac{|v(\xi)|^2  }{(\omega(\xi)+E)^2}\ud t\ud\xi \stackrel{\eqref{eq:integral-E}}{\lesssim}  |p|^\gamma E^{\delta/\gamma-1}.
 %
\end{align}
This proves that the limit $\Lambda \to \infty $ exists, and satisfies
\begin{equation}
 |\theta_{1,0}(p,E)| \lesssim H_0(p,E) E^{-(1-\delta/\gamma)},
\end{equation}
so $\theta_{1,0}\in \cK_{0,(1-\delta/\gamma)}$.

To prove the bounds on the derivatives, note that
\begin{equation}
  \theta_{1,0}(p,E+h)-\theta_{1,0}(p,E) = 
 \int_E^{E+h} \int\limits_{\R^d}   \frac{|v(\xi)|^2}{(\Omega(p+\xi) + \omega(\xi) +t )^2}\ud \xi\ud t,
\end{equation}
which is absolutely convergent. From this we easily deduce that the derivative in $E$ exists and satisfies the bound
\begin{equation}
 |\partial_E \theta_{1,0}(p,E)|\lesssim E^{-(1-\delta/\gamma)} \leq H_0(p,E) E^{-1-(1-\delta/\gamma)},
\end{equation}
as claimed.

For the first derivative in $p$ we also first take the difference and then use Taylor expansion once, to obtain
\begin{align}
 &\theta_{1,0}(p+h,E) - \theta_{1,0}(p,E) \notag \\
 &= \lim_{\Lambda \to \infty} \int_{|\xi|\leq \Lambda} \int_{0}^{1} \bigg\langle h, \nabla_q\frac{|v(\xi)|^2}{\Omega(\xi+q) + \omega(\xi)+E }\Big\vert_{q=p+th}\bigg\rangle \ud t \ud \xi  \\
 &=  \int_{\R^d} \int_{0}^{1}\int_0^1 \bigg\langle h, \nabla_q^2\frac{|v(\xi)|^2}{\Omega(\xi+q) + \omega(\xi)+E }\Big\vert_{q=s(p+th)} (p+th)\bigg\rangle \ud s \ud t \ud \xi,\notag
 \end{align}
 where we have used rotation invariance and the bound on the Hessian as above.
 From this, we see that $\partial^{\nu}_p\theta_{1,0}$ exists for $|\nu|=1$ and 
 \begin{align}
  |\partial^{\nu}_p\theta_{1,0}(p,E)| \lesssim |p| E^{\delta/\gamma-1} \leq H_0(p,E) E^{-1/\gamma -(1-\delta/\gamma)}.
 \end{align}
For the second derivative one proceeds in the same way, writing the difference of derivatives as
 \begin{align}
  &\partial_j\theta_{1,0}(p+h,E) - \partial_j\theta_{1,0}(p,E) \notag \\
 &=\int_{\R^d} \int_{0}^{1} \bigg\langle e_j, \nabla_q^2\frac{|v(\xi)|^2}{\Omega(\xi+q) + \omega(\xi)+E }\Big\vert_{q=p+th} h\bigg\rangle  \ud t \ud \xi,
 \end{align}
 which implies the claim by the same arguments as above.
\end{proof}

 \begin{thm}\label{thm:theta}
  Assume the hypothesis~\ref{ass:main} and let $0\leq m\leq n\leq n_*$. Then for $\Lambda\in \R_+\cup \{\infty\}$ :
  
  \begin{enumerate}[a)]
   \item\label{thm:theta plain} The operator $\theta_{\Lambda, n,m}$ defined by~\eqref{eq:theta-nm-def},~\eqref{eq:theta-n0-def} satsifies
   \begin{equation*}
   \theta_{\Lambda,n,m} \in \cK_{m, (n-m)(1-\delta/\gamma)}, 
  \end{equation*}
where the bounds required by Definition~\ref{def:K_n} hold uniformly in $\Lambda$;
\item\label{thm:theta rot} The kernel $\theta_{\Lambda, n, m}(Q,R,p,E)$ is invariant under simultaneous rotations of its arguments in $\R^d$, i.e, the maps
\begin{equation}
 Q,R,p \mapsto Oq_1, \dots, Or_m, Op,
\end{equation}
for $O\in SO(d)$.
\item\label{thm:theta deriv} The kernel $\theta_{\Lambda, n, m}(Q,R,p,E)$ is continuously differentiable in $E$, $\gamma$-times differentiable in $p$, and satisfies
\begin{align*}
 \partial_E \theta_{\Lambda, n, m} &\in \cK_{m, 1+(n-m)(1-\delta/\gamma)} \\
 \partial_p^\nu \theta_{\Lambda, n, m} &\in \cK_{m, |\nu|/\gamma+(n-m)(1-\delta/\gamma)} , \qquad |\nu|\leq \gamma,
\end{align*}
uniformly in $\Lambda$;
\item\label{thm:theta conv} The kernels $\theta_{\Lambda, n, m}$ and their derivatives converge pointwise to $\theta_{n,m}$ and its derivatives as $\Lambda \to \infty$.
  \end{enumerate}
\end{thm}

Note that $\theta_{\Lambda,  n,m}$ is a sum of terms $\vartheta_{\Lambda, J, I, L}$ that  
  are essentially given as $(a(v) H_0^{-1} \star_{\ell_0} \tau_{\Lambda, I,L_1^{\nu-1}}) H_0^{-1}\star_{\ell_\nu}a^*(v)$, where $\tau_{\Lambda, I,L_1^{\nu-1}}$ is defined by~\eqref{eq:kernel-recursion} as a $\star$-product involving only the factors $H_0^{-1}$ and $\theta_{\Lambda,j,i}$, $i\leq j\leq n$ (thereby $\tau_{\Lambda, I,L_1^{\nu-1}}$ depends implicitly on $J$).
  Before proving Theorem~\ref{thm:theta} we thus prove a Lemma on  $\tau_{\Lambda, I,L_1^{\nu-1}}$ that proves the corresponding statement for $\tau$ and also plays a role in the proof of Proposition~\ref{prop:R}.

  \begin{lem}\label{lem:tau}
  Assume the hypothesis~\ref{ass:main}.
   Let $1\leq \nu\leq n\leq n_*$, $J\in \{0, \dots, n\}^\nu$ with $|J|<n_*$, $I\in \N_0^\nu$ with $i_\mu \leq j_\mu$ for $\mu=1, \dots, \nu$. For $\nu=1$ set $L=\emptyset$ and for $\nu>1$ let $L\in \N_0^{\nu-1}$ satisfying 
   \begin{equation*}
    \ell_\mu \leq \min \Big\{|I_1^\mu|-|L_1^{\mu-1}|, i_{\mu+1} \Big\},   \qquad  \mu=1, \dots, \nu-1.
   \end{equation*}
Assume that the statement of Theorem~\ref{thm:theta} holds for all $\theta_{\Lambda, j, i}$ for $0\leq i\leq j\leq n$ and $\Lambda \in \R_+ \cup \{\infty\}$.

For $\Lambda \in \R_+\cup\{\infty\}$ and $\tau_{\Lambda, I,L}$ defined by~\eqref{eq:kernel-recursion} we then have
\begin{enumerate}[a)]
\item\label{lem:tau rot} The kernel $\tau_{\Lambda, I,L}$ is invariant under simultaneous rotations of $Q,R,p \in \R^{d(2m+1)}$, $m=|I|-|L|$;
\item\label{lem:tau diff} The kernel $\tau_{\Lambda, I,L}(Q,R,p,E)$ is continuously differentiable in $E$ and $\gamma$-times differentiable in $p$;
%
\item\label{lem:tau K} For all $0\leq |\mu|\leq \gamma$ and uniformly in $\Lambda$
\begin{align*}
\partial_p^\mu \tau_{\Lambda, I, L} &\in \cK_{|I|-|L|,|\mu|/\gamma+\sigma}, \\
 \partial_E \tau_{\Lambda, I, L} &\in \cK_{|I|-|L|,1+\sigma}
\end{align*}
with
\begin{equation*}
 \sigma:=(|J|-|I|+|L|)(1-\delta/\gamma);
\end{equation*}
\item\label{lem:tau conv} For $\Lambda\to \infty$, $\tau_{\Lambda, I, L}$ and its derivatives converge pointwise to $\tau_{I,L}$ and its derivatives.
\end{enumerate}

\end{lem}
\begin{proof}
 We proceed by (finite) induction on $\nu=1, \dots, n$. For $\nu=1$, we have $J=(j_1)$, $L=\emptyset$ and $\tau_{\Lambda,i_1, \varnothing}=\theta_{\Lambda,j_1, i_1}$ satisfies all the claimed properties since the statement of Theorem~\ref{thm:theta} was assumed to hold and $j_1=|J|<n_*$.
 
 For part~\ref{lem:tau rot}), observe that changing variables $\xi_1, \dots , \xi_\ell \mapsto O\xi_1, \dots, O\xi_\ell$ in the formula~\eqref{eq:kappa-star_ell} shows that the $\star_\ell$-product preserves the property of being invariant under a simultaneous rotation of all arguments.
 
 Part~\ref{lem:tau diff}) will be proved together with the bound on the derivatives of part~\ref{lem:tau K}).
 Assume that the statement of the Lemma holds for $1\leq \nu-1<n_*$.
 for $J=(J_1^{\nu-1}, j_\nu)$  we have by Theorem~\ref{thm:K-star-prod} (taking $L_1^{\nu-2}=\emptyset$ if $\nu=2$)
 \begin{equation}\label{eq:tau Inot0 rec}
  \tau_{\Lambda, I , L} = \tau_{\Lambda, I_1^{\nu-1}, L_1^{\nu-2}}H_0^{-1} \star_{\ell_{\nu-1}} \theta_{\Lambda, j_\nu, i_\nu}\in \cK_{m, \sigma},
 \end{equation}
 with 
 \begin{equation}\label{eq:tau step}
  \begin{aligned}
 m&= |I_1^{\nu-1}| - |L_1^{\nu-2}| + i_{\nu}- \ell_{\nu-1},  \\
 \sigma &= (|J_1^{\nu-1}|-|I_1^{\nu-1}| + |L_1^{\nu-2}| + j_\nu - i_\nu + \ell_{\nu-1})(1-\delta/\gamma),
  \end{aligned}
 \end{equation}
due to the induction hypothesis and the assumed properties of $\theta_{\Lambda, j_\nu, i_\nu}$.
This shows part~\ref{lem:tau K}) for $\mu=0$.

Now let $|\mu|=1$. By Assumption~\ref{ass:main} we have
\begin{equation}\label{eq:res p-deriv}
\big|(H_0 \partial_{p_i} H_0^{-1})(p,E)\big| =\bigg|\frac{\partial_{p_i} \Omega(p)}{\Omega(p)+E} \bigg| 
\lesssim \frac{\Omega(p)^{1-1/\gamma}}{\Omega(p)+E} \leq \frac{1}{(\Omega(p)+E)^{1/\gamma}}.
\end{equation}
It then follows from Lemma~\ref{lem:star ell=0} (in case $I_1^{\nu-1}\neq 0$, otherwise the product is just the pointwise product and the analogous bound is is trivial) that for all 
\begin{equation}
0\leq  \lambda < \sigma_{\nu-1}:=  (|J_1^{\nu-1}|-|I_1^{\nu-1}| + |L_1^{\nu-2}|)(1-\delta/\gamma)
\end{equation}
 we have, setting $t=s+1/\gamma$,
\begin{align}
 &\Big|\tau_{\Lambda, I_1^{\nu-1}, L_1^{\nu-2}} (\partial_{p_i} H_0^{-1})H_0 (Q,R,p,E)\Big|  \\
 &\lesssim \frac{1}{(E+\Omega(Q))^{1/\gamma}} \min_{s\in [\lambda-1, 1-\lambda]}  \rho_{|I_1^{\nu-1}| - |L_1^{\nu-2}|,\lambda+s}(Q,E)\tilde\rho_{|I_1^{\nu-1}| - |L_1^{\nu-2}|,\lambda-s}(R,E). \notag \\
 &\leq \min_{t\in [\lambda+1/\gamma-1, 1-\lambda+1/\gamma]}  \rho_{|I_1^{\nu-1}| - |L_1^{\nu-2}|,\lambda+t}(Q,E)\tilde\rho_{|I_1^{\nu-1}| - |L_1^{\nu-2}|,\lambda+1/\gamma-t}(R,E).\notag
\end{align}
Hence 
\begin{equation}
 \tau_{\Lambda, I_1^{\nu-1}, L_1^{\nu-2}} (\partial_{p_i} H_0^{-1})H_0 \in \cK_{|I_1^{\nu-1}| - |L_1^{\nu-2}|, 1/\gamma+\sigma_{\nu-1}},
\end{equation}
and by the induction hypothesis $\partial_{p_i}\tau_{\Lambda, I_1^{\nu-1}, L_1^{\nu-2}} $ is an element of the same class.
Thus by Theorem~\ref{thm:K-star-prod}
\begin{align}\label{eq:tau p1-deriv}
 (\partial_{p_i} \tau_{\Lambda, I_1^{\nu-1}, L_1^{\nu-2}} H_0^{-1})& \star_{\ell_{\nu-1}} \theta_{\Lambda, j_\nu, i_\nu} \notag \\
 &+ \tau_{\Lambda, I_1^{\nu-1}, L_1^{\nu-2}} H_0^{-1}\star_{\ell_{\nu-1}} \partial_{p_i}\theta_{\Lambda, j_\nu, i_\nu} \in \cK_{m, \sigma + 1/\gamma}, 
\end{align}
with $m, \sigma$ as in~\eqref{eq:tau step}. Since the bounds on the kernels are uniform in $p$ and the internal integrals absolutely convergent (cf.~Lemmas~\ref{lem:star ell<n},~\ref{lem:star ell=n}), this shows that $\tau_{\Lambda, I , L}$ is continuously differentiable in $p$, with derivative given by~\eqref{eq:tau p1-deriv}. For the derivatives in $p$ of order $|\mu|=2=\gamma$ and the derivative in $E$ the claim follows from the same argument, in view of the inequalities
\begin{align}\label{eq:res p-Hess}
\begin{aligned}
 \big|(H_0 \partial^\mu_p H_0^{-1})(p,E)\big| &\lesssim H_0(p,E)^{-1}, \qquad |\mu|=\gamma, \\
 \big|(H_0 \partial_E H_0^{-1})(p,E)\big|&\lesssim H_0(p,E)^{-1}.
 \end{aligned}
 \end{align}
 This proves~\ref{lem:tau diff}) and~\ref{lem:tau K}).

Concerning part~\ref{lem:tau conv}), the uniform bounds on $\tau_{\Lambda, I_1^{\nu-1}, L_1^{\nu-2}}$, and $\theta_{\Lambda, j_\nu, i_\nu}$ together with the fact that the integrals in $\star_\ell$ converge absolutely (cf.~Lemmas~\ref{lem:star ell<n},~\ref{lem:star ell=n}) imply that  $\tau_{\Lambda, I, L}$ converges pointwise to $\tau_{I, L}$ for $\Lambda \to \infty$, by dominated convergence. The same holds for the derivatives in $E$, $p$.
This completes the proof of the Lemma.
\end{proof}

\begin{proof}[Proof of Theorem~\ref{thm:theta}]
 We proceed by induction on $n$.
 
\paragraph{Base case $n=1$} 
The base case for $\theta_{\Lambda, 1,0}$ is established in Lemma~\ref{lem:theta_10}. The kernel 
$\theta_{\Lambda,1,1}$ is invariant under the rotation $(q,r,p) \mapsto (Oq, Or, Op)$, $O\in SO(d)$, by rotation invariance of $v, \omega, \Omega$ (and the sharp UV-cutoff), so~\ref{thm:theta rot}) holds. We have
 \begin{equation}
  |\theta_{\Lambda,1,1}(q,r,p,E)| = \frac{|v_\Lambda(q)v_\Lambda(r)|}{\Omega(p+q+r)+\omega(q)+\omega(r) + E } \leq \min_{s\in [-1,1]}  \rho_{1,s}(q,E) \tilde\rho_{q, -s}(r,E),
 \end{equation}
so $\theta_{\Lambda, 1, 1}\in \cK_{1,0}$ uniformly in $\Lambda$, and~\ref{thm:theta plain}) holds for $\theta_{\Lambda, 1,1}$.
Moreover, as $E\geq 1$ we have for all $0\leq \lambda \leq 1$
\begin{align}
 |\partial_E \theta_{\Lambda,1,1}(q,r,p,E)| &= \frac{|v_\Lambda(q)v_\Lambda(r)|}{(\Omega(p+q+r)+\omega(q)+\omega(r) + E )^2} \notag\\
 &\leq  \min_{s\in [\lambda-1,1-\lambda]}  \rho_{1,\lambda+s}(q,E) \tilde\rho_{1, \lambda-s}(r,E),
\end{align}
so $\partial_E \theta_{\Lambda, 1, 1}\in \cK_{1,1}$ uniformly in $\Lambda$.
For the $p$-derivatives, we use the inequality~\eqref{eq:res p-deriv}
to obtain for all $0\leq \lambda \leq 1/\gamma$
\begin{align}
 |\nabla_p \theta_{\Lambda,1,1}(q,r,p,E)| 
 %
 &\lesssim \frac{|v(q)v(r)|}{(\Omega(p+q+r)+\omega(q)+\omega(r) + E )^{1+1/\gamma}} \notag \\
 &\leq  \min_{s\in [\lambda-1, 1-\lambda]} \rho_{1,\lambda+s}(q,E)\tilde \rho_{1,\lambda-s}(r,E),
 \end{align}
 whence $\partial^\nu_p \theta_{\Lambda,1,1}\in \cK_{1,1/\gamma}$ for $\nu=1$.
If $\gamma=1$ this already shows~\ref{thm:theta deriv}). If $\gamma=2$ calculating the second derivative and bounding it in the same way yields~\ref{thm:theta deriv}). 

We have convergence $\theta_{\Lambda,1,1}(q,r,p,E)\to \theta_{1,1}(q,r,p,E)$ for $\Lambda\to \infty$ by convergence of $v_\Lambda \to v$ and thus~\ref{thm:theta conv}).
This establishes the base case $n=1$.

\paragraph{Induction step for $m=0$}
To perform the induction step for $\theta_{n+1, 0}$ we need to take the limit $\Lambda \to \infty$, so we need to consider $\Lambda<\infty$ first. It is sufficient to prove the claim for each summand in~\eqref{eq:theta-n0-def}, i.e. for 
\begin{equation}
\lim_{\Lambda \to \infty} (\vartheta_{\Lambda, J, I, L}- E_{\Lambda, J, I, L})
\end{equation}
where $\vartheta_{\Lambda, J, I, L}$ is given by~\eqref{eq:vartheta-def} with $\nu\leq n$, $J\in \{1, \dots, n\}^\nu$ with $|J|=n$, $I\in \N_0^\nu$ with $i_\mu\leq j_\mu$, $1\leq \mu\leq \nu$ and  $L=(\ell_1, \dots, \ell_\nu)$ is chosen according to~\eqref{eq:L-constr-0}.

There are two somewhat distinct cases, $I=0$ and $I\neq 0$, to be considered. 
For $I=0$, we have $\ell_\nu=1$, $L_1^{\nu-1}=0$, and
\begin{align}
 &\vartheta_{\Lambda, J, 0, (0,\dots, 1)}- E_{\Lambda, J, 0, (0,\dots, 1)} \\
 &= 
 \int_{|\xi|\leq \Lambda} 
 \begin{aligned}[t]
  \frac{|v_\Lambda(\xi)|^2 \tau_{\Lambda, 0,0}(p+\xi,E+\omega(\xi))}{(\Omega(p+\xi)+ \omega(\xi)+ E  )^2} 
 -\frac{|v_\Lambda(\xi)|^2 \tau_{\Lambda, 0,0}(\xi,\omega(\xi)+E_0)}{(\Omega(\xi)+ \omega(\xi)+E_0  )^2} \ud \xi.
  \end{aligned}\notag
 \end{align}
 In view of properties of $\tau_{\Lambda, 0, 0}(p,E)$ from Lemma~\ref{lem:tau} the limit $\Lambda\to \infty$ can be treated in complete analogy to Lemma~\ref{lem:theta_10} by Taylor-expanding the integrand (appealing to rotation invariance to eliminate the first-order term if necessary, and using that $(n+1)(1-\delta/\gamma)\leq 1$).
 
For $I\neq 0$ we have $\ell_0=\ell_\nu=1$, and  $|I|-|L_1^{\nu-1}|=1$ by~\eqref{eq:L-constr-0}.
Then
\begin{align}\label{eq:theta I-not0}
 &\vartheta_{\Lambda, J, I, L} - E_{\Lambda, J, I, L} \\ &= 
 \begin{aligned}[t]
 \int\limits_{|\xi|\leq \Lambda}\int\limits_{|\eta|\leq \Lambda} & \frac{v(\xi) v(\eta) \tau_{\Lambda, I,L_1^{\nu-1}}(\xi, \eta,p,E)}{(\Omega(p+\xi)+ \omega(\xi)+ E  )(\Omega(p+\eta)+ \omega(\eta)+ E  )} \\
 &-\frac{v(\xi) v(\eta) \tau_{\Lambda, I,L_1^{\nu-1}}(\xi, \eta,0,E_0)}{(\Omega(\xi)+ \omega(\xi)  )(\Omega(\eta)+ \omega(\eta) +E_0 )} \ud \xi \ud \eta.
\end{aligned}\notag
\end{align}
Invariance of $\theta_{\Lambda, n+1,0}$ under rotations $p\mapsto Op$ follows from this formula by changing variables $\xi, \eta \mapsto O\xi, O\eta$ and using the rotation invariance of $\tau_{\Lambda, I,L_1^{\nu-1}}$, $\Omega$, $\omega$, and $v$.

If we express this using Taylor expansion as in Lemma~\ref{lem:theta_10} (for $\gamma=2$, for $\gamma=1$ we expand only once in $p$), we find
\begin{align}
 \vartheta_{\Lambda, J, I, L} - E_{\Lambda, J, I, L} 
 & =\int\limits_{E_0}^E \int\limits_{|\xi|\leq \Lambda}\int\limits_{|\eta|\leq \Lambda} \partial_t  F(\xi, \eta, 0, t) \ud t \ud \xi \ud \eta \label{eq:theta_n0 E}\\
 &\qquad + \int\limits_0^1 \int\limits_{|\xi|\leq \Lambda}\int\limits_{|\eta|\leq \Lambda}
  \Big\langle p,\nabla_q F(\xi, \eta, q, E) \Big\vert_{q=0}  \bigg\rangle \label{eq:theta_n0 rot}\\
 &\qquad +\int\limits_0^1 (1-t)\int\limits_{|\xi|\leq \Lambda}\int\limits_{|\eta|\leq \Lambda}
 \ \Big\langle p,\nabla_q^2 F(\xi, \eta, q, E) \Big\vert_{q=tp} p \bigg\rangle \label{eq:theta_n0 Hess}
\end{align}
with 
\begin{equation}
 F(\xi, \eta, p, E)=\frac{v(\xi) v(\eta) \tau_{\Lambda, I,L_1^{\nu-1}}(\xi, \eta,p,E)}{(\Omega(p+\xi)+ \omega(\xi)+ E  )(\Omega(p+\eta)+ \omega(\eta)+ E  )}.
\end{equation}
By Lemma~\ref{lem:tau}~\ref{lem:tau K}) with $|J|=n$, $|I|-|L_1^{\nu-1}|=1$, we have for 
$\lambda< (n-1)(1-\delta/\gamma)<1$,
\begin{align}
 &\bigg|\int\limits_{E_0}^E \int\limits_{|\xi|\leq \Lambda}\int\limits_{|\eta|\leq \Lambda} \partial_t  F(\xi, \eta, 0, t) \ud t \ud \xi \ud \eta\bigg| \notag\\
 &\lesssim \int_0^E t^{-\lambda} \bigg(\int_{\R^d} \frac{|v(\xi)^2}{(t+\omega(\xi))^2} \ud \xi\bigg)^2 \ud t \notag\\
 &\stackrel{\eqref{eq:integral-E}}{\lesssim} \int_0^E t^{-\lambda-2(1-\delta/\gamma)} \ud t \lesssim E^{1-\lambda -2(1-\delta/\gamma) },
 \label{eq:theta_n0 log}
\end{align}
as $\lambda + 2(1-\delta/\gamma) < (n+1)(1-\delta/\gamma)\leq 1$ for $n+1\leq n_*$.

The term~\eqref{eq:theta_n0 rot} vanishes due to rotation invariance of $\Omega, \omega, v$, and $\tau$, as argued in Lemma~\ref{lem:theta_10}. 

To bound~\eqref{eq:theta_n0 Hess} we use Lemma~\ref{lem:tau}~\ref{lem:tau K}) to obtain with $\lambda $ as above (keeping in mind that $\gamma=2$)
\begin{align}
 &\bigg| \int\limits_0^1 (1-t)\int\limits_{|\xi|\leq \Lambda}\int\limits_{|\eta|\leq \Lambda}
 \ \Big\langle p,\nabla_q^2 F(\xi, \eta, q, E) \Big\vert_{q=tp} p \bigg\rangle\bigg| \notag \\
 &\lesssim p^2 E^{-\lambda}\int\limits_0^1 (1-t) \bigg(\int_{\R^d} \frac{|v(\xi)|^2}{(\omega(\xi)+E)^2}\ud \xi \bigg) \notag \\
 &\lesssim |p|^\gamma E^{-\lambda -2(1-\delta/\gamma)}.
\end{align}

The limit of for $\Lambda\to \infty$ thus exists and satisfies the required bound.
The bounds on the derivatives are obtained by arguing as in Lemma~\ref{lem:theta_10}. This proves all claims for $\theta_{\Lambda, n+1, 0}$.

\paragraph{Induction step for $m>0$}
Now consider $\theta_{\Lambda, n+1, m}$ for $m>0$. The summand $\vartheta_{J,I, L}$ is given by
\begin{equation}
 \vartheta_{\Lambda, J, I, L} =\Big( a(v_\Lambda) H_0^{-1}\star_{\ell_0} \tau_{\Lambda, I,L_1^{\nu-1}} \Big) H_0^{-1} \star_{\ell_{\nu}}  a^*(v_\Lambda),
\end{equation}
where 
\begin{equation}
 m=|I|- |L_1^{\nu-1}| +1 -\ell_0-\ell_1 = |I|- |L| +1
\end{equation}
is the number of uncontracted variables.

In the case $I\neq 0$, the claim of  part~\ref{thm:theta plain}) follows immediately from Lemma~\ref{lem:tau}~\ref{lem:tau K}) and Lemma~\ref{lem:a-star}.

The case $I=0$ only occurs for $m=1$ and with $\ell_0=\ell_1=0$. Using Lemma~\ref{lem:tau}~\ref{lem:tau K}) we obtain in this case for $\lambda<n(1-\delta/\gamma)$
\begin{align}
 |\vartheta_{\Lambda, J, I, L}(q,r,p,E)| &=
 \bigg|\frac{v_\Lambda(q)v_\Lambda(r)\tau_{\Lambda,0, 0}(p+q+r, E + \omega(q)+\omega(r)}{H_0(p+q+r, E + \omega(q)+\omega(r))^2}\bigg| \notag\\
 & \lesssim \frac{|v(q)| |v(r)|}{(E+ \omega(q)+\omega(r))^{1+\lambda}},
\end{align}
which clearly satisfies~\ref{thm:theta plain}) with $n+1-m=n$.

Rotation invariance of these kernels, part~\ref{thm:theta rot}), follows from a change of variables in the formulas of Lemma~\ref{lem:a-star}, as in Lemma~\ref{lem:tau}~\ref{lem:tau rot}).

Part~\ref{thm:theta deriv}) is proved in the same way as Lemma~\ref{lem:tau}~\ref{lem:tau K}) for the derivatives, making use of the bounds on derivatives of $\tau_{I,L_1^{\nu-1}}$ and Lemma~\ref{lem:a-star}.

Moreover, we have pointwise convergence of $\vartheta_{\Lambda, J, I, L}$ to $\vartheta_{J, I, L}$ as $\Lambda \to \infty$ by our uniform bounds and the fact that the integrals in Lemma~\ref{lem:a-star} converge absolutely. This completes the proof.
\end{proof}

\subsection{Proof of Propositions~\ref{prop:T} and~\ref{prop:R}}\label{sect:propT}


\begin{proof}[Proof of Proposition~\ref{prop:T}]
\textit{a)} We need to show that $T_{\Lambda,n}=\sum_{m=1}^n \theta_{\Lambda,n,m}$ is bounded relative to $H_0 \ud \Gamma(\omega)^{-s}$ for, $s<\tfrac12 n(1-\delta/\gamma)$, uniformly in $\Lambda \in \R_+ \cup \{\infty\}$.

For $m=0$, we have from Theorem~\ref{thm:theta}~\ref{thm:theta plain}) the bound
\begin{equation}
 |\theta_{n,0}(\ud \Gamma(k)-p, \ud \Gamma(\omega) + E_0) |\lesssim H_0 (\ud \Gamma(\omega) +E_0)^{-s},
\end{equation}
for $s<n(1-\delta/\gamma)$, which is better than required.

For $m>0$, we have the bound of Theorem~\ref{thm:theta}~\ref{thm:theta plain}), which together with Lemma~\ref{lem:K_n-op} proves that for $\Psi \in D(H_0)$
\begin{equation}
\|\theta_{n,m}\Psi\| \lesssim \| (\ud \Gamma(\omega)+E_0)^{1-s} \Psi \|
\end{equation}
for $s<\tfrac12 n(1-\delta/\gamma)$. This proves the bound of Proposition~\ref{prop:T}~\ref{prop:T-bound}).
 
\textit{b)}  We need to prove strong convergence of $T_{\Lambda, n}$ to $T_n$  in the sense of bounded operators from $D(H_0)$ to $\cF$. Since by~\ref{prop:T-bound}) the familiy $T_\Lambda$ is uniformly bounded in $\mathscr{L}(D(H_0), \cF)$ it is sufficient to prove this on a dense subset $D\subset D(H_0)$.
As this set we choose the elements $\Psi\in \cF$ that are finite linear combinations of compactly supported functions in $L^2(\R^{d k})$, $k\geq 0$. 
The pointwise convergence of the kernels $\theta_{\Lambda,n,m}$ to $\theta_{n,m}$, $m\leq n$, established in Theorem~\ref{thm:theta} implies that for $\Psi \in D$, we have
\begin{equation}
 \lim_{\Lambda \to \infty} \big(\theta_{\Lambda, n,m}\Psi\big)^{(k)}(K) =\big(\theta_{n,m}\Psi\big)^{(k)}(K) 
\end{equation}
pointwise in $K\in \R^{dk}$. For $m=0$, convergence in $L^2(\R^{dk})$ follows immediately from the uniform bound on the multiplication operator $\theta_{\Lambda,n,0}$ and dominated convergence. 
Moreover, by the bound of Theorem~\ref{thm:theta}~\ref{thm:theta plain}), we have  that, for $k\geq m>0$, $|\theta_{\Lambda, n,m}\Psi^{(k)}(K)|$ is (up to symmetrisation and a numerical prefactor) given by
\begin{align}
  %
  &  \int\limits_{\R^{dm}} \Big|\theta_{\Lambda, n,m}\big(K_1^m,\Xi,\omega(K_{m+1}^k)+E_0\big)\Psi\uppar{k} (\Xi, K_{m+1}^k) \Big|\ud \Xi\notag \\
  & \lesssim \rho_{m,1}(K_1^m,\omega(K_{m+1}^k)+E_0) \int\limits_{\R^{dm}} \tilde \rho_{m,-1}(\Xi,\omega(K_{m+1}^k)+E_0)\big|\Psi\uppar{k} (\Xi, K_{m+1}^k)\big|\ud \Xi.
  \notag 
\end{align}
This expression has compact support in $K_{m+1}^k$ since $\Psi\uppar{k}$ has compact support, and it is square integrable in $K_1^m$ because $\rho_{m,1}(K_1^m, E_0)$ is square integrable. Thus $\theta_{\Lambda, n,m}\Psi^{(k)}$ converges to $\theta_{n,m}\Psi^{(k)}$ in $L^2(\R^{dk})$ by dominated convergence. This proves part~\ref{prop:T-conv}) of Proposition~\ref{prop:T}.

\textit{c)} The fact that $T_n$ is well defined on $D(H_0)$ follows from part~\ref{prop:T-bound}).
By construction, $T_{\Lambda,n}$ is symmetric. Indeed, the formula~\eqref{eq:T_Lambda def} shows inductively that $T_{\Lambda, n}$ takes the form $a(v_\Lambda) S_\Lambda a^*(v_\Lambda)$ with a symmetric operator $S_\Lambda$. Hence by the convergence proved in part~\ref{prop:T-conv}) $T_n$ is also symmetric on $D(H_0)$.
\end{proof}

We now turn to the remainder $R_\Lambda$, $\Lambda \in \R_+ \cup \{\infty\}$.

\begin{proof}[Proof of Proposition~\ref{prop:R}]
 
 For $\Lambda \in \R_+ \cup \{\infty\}$ the operator $R_\Lambda$ is given as 
 \begin{equation*}
  R_\Lambda=G_{T_\Lambda}^*  \sum_{m=0}^{n_*-1} \sum_{j=m+1}^{n_*} T_{\Lambda,j} \sum_{\nu=1}^{n_*-1-m}(-1)^{\nu}\sum_{J\subset\{1, \dots, n_*\}^\nu \atop |J|=n_*-1-m} \Big( \prod_{\mu=1}^\nu H_0^{-1} T_{\Lambda, j_\mu}\Big) G_0.
 \end{equation*}
 We first prove uniform bounds on $G_0, G_{T_\Lambda}^*$ where we gain some decay. Then we show a (form) bound on the operator between $G_0$, $G_{T_\Lambda}^*$  that allows us to conclude boundedness of $R_\Lambda$. 
 
 First, $G_0^* \ud \Gamma(\omega)^{s}$ is bounded for $1-s>\tfrac12 (1+\delta/\gamma)$, i.e., $s<\tfrac12 (1-\delta/\gamma)$, by Lemma ~\ref{lem:a-omega}, since
 \begin{equation}
   \| a(v)H_0^{s-1}\| \leq \|v\omega^{s-1}\|_{L^2}.
 \end{equation}
Hence $\ud \Gamma(\omega)^s G_0$ is also bounded. For $G_{T_\Lambda}^*$ we find
\begin{align}
 \|G_{T_\Lambda}^*\ud \Gamma(\omega)^{s}\| &\leq  
 \|a(v)  H_0^{-1} \ud\Gamma(\omega)^{s}- a(v)(H_0+T_\Lambda)^{-1} T_\Lambda H_0^{-1}\ud \Gamma(\omega)^{s}\| \notag \\
 & \leq \|G_0^* \ud\Gamma(\omega)^{s}\| + \|G_{T_\Lambda}\| \| T_\Lambda H_0^{-1}\ud \Gamma(\omega)^{s}\|,
\end{align}
where the last norm is finite for $s<\tfrac12(1-\delta/\gamma)$ by Proposition~\ref{prop:T}~\ref{prop:T-bound}).
Similarly, we have for $t>\tfrac12(1+ \delta/\gamma)$
 \begin{equation}
  \| G_0^* H_0 \ud \Gamma(\omega)^{-t}\| \leq \|v\omega^{-t}\|_{L^2}
 \end{equation}
by Lemma~\ref{lem:a-omega}.

We now consider the factor between $G_{T_\Lambda}^*$, $G_0$. Fix one summand, i.e., $j=j_0\in \{m+1, n_*-1\}$, $J\subset \{1, \dots, n_*\}$ with $|J| = n_*-s-m$. We can then expand each $T_{\Lambda,j_\mu}$, $\mu\in \{0, \dots, \nu\}$ into 
\begin{equation}
T_{\Lambda, j_\mu}=\sum_{i_\mu=1}^{j_\mu} \theta_{\Lambda, j_\mu, i_\mu} 
\end{equation}
and the operator products into the kernel products $\star_\ell$ with all possible choices of $\ell$, as in Section~\ref{sect:algorithm}.
After this, a fixed summand takes the form
\begin{equation}
 (-1)^{\nu}\theta_{\Lambda, j_0, i_0} H_0^{-1} \star_{\ell_0} \tau_{\Lambda, I, L}. 
\end{equation}
Note that $|J|<n_*$ so Lemma~\ref{lem:tau} applies to $\tau_{\Lambda, I, L}$. 
We then obtain from  Theorem~\ref{thm:K-star-prod} with Theorem~\ref{thm:theta}
\begin{equation}
 \theta_{\Lambda, j_0, i_0} H_0^{-1} \star_{\ell_0} \tau_{\Lambda, I, L} \in \cK_{k, \lambda}
\end{equation}
with 
\begin{equation}
\begin{aligned}
\lambda &= (j_0 + n_*-1-m - k)(1-\delta/\gamma), \\
 k&= i_0 + |I| - \ell_0 - |L|.
\end{aligned}
\end{equation}

We have $k=0$ only if $i_0=0$ and $I=0$. Then, by choice of $n_*$, we have $n_*(1-\delta/\gamma)>\delta/\gamma$ so there exists $s<\tfrac12(1-\delta/\gamma)$, $t>\tfrac12(1+\delta/\gamma)$ with $t-s-n_*(1-\delta/\gamma)< 0$. Thus
\begin{align}
&\| G_{T_\Lambda}^*  \theta_{\Lambda, j_0, 0} H_0^{-1} \tau_{\Lambda, 0, 0} G_0\| \notag\\
&\leq \| G_{T_\Lambda}^* \ud \Gamma(\omega)^s\| \| \ud \Gamma(\omega)^{-s} \theta_{\Lambda, j_0, 0} H_0^{-1} \tau_{\Lambda, 0, 0} \ud \Gamma(\omega)^{t} H_0^{-1}\|  \| \ud \Gamma(\omega)^{-t} H_0 G_0 \| \notag \\
&\lesssim \| G_{T_\Lambda}^* \ud \Gamma(\omega)^s\| \| \ud \Gamma(\omega)^{-s+t-n_*(1-\delta/\gamma)} \| \ud \Gamma(\omega)^{-t} H_0 G_0 \|
\end{align}
is bounded. 

For $k\geq 1$, an operator in $\cK_{k, \lambda}$ 
is form-bounded relative to $\ud \Gamma(\omega)^{2s}$ if $2s>1-\lambda-k(1-\delta/\gamma)$ by Lemma~\ref{lem:K_n-form}. Since $j+n_*-m-1\geq n_*$, we can choose $s>\tfrac12(1-\delta/\gamma)$ so that this holds. Consequently,
\begin{align}
 &\| G_{T_\Lambda}^* \theta_{\Lambda, j_0, i_0} H_0^{-1} \star_{\ell_0} \tau_{\Lambda, I, L} G_0\| \notag\\
 &\leq \| G_{T_\Lambda}^* \ud \Gamma(\omega)^s\| \| \ud \Gamma(\omega)^{-s} \theta_{\Lambda, j_0, i_0} H_0^{-1} \star_{\ell_0} \tau_{\Lambda, I, L}  \ud \Gamma(\omega)^{-s}\|  \| \ud \Gamma(\omega)^{s} G_0 \|
\end{align}
is bounded uniformly in $\Lambda$.
This shows that $\R_\Lambda$, $\Lambda \in \R_+\cup\{\infty\}$ is uniformly bounded.

Strong convergence of $R_\Lambda$ to $R$ for $\Lambda \to \infty$ follows from these bounds together with the convergence $G_{T_\Lambda}\to G_T$, and the convergence of the kernels $\theta_{\Lambda, n,m}\to \theta_{m,n}$ proved in Theorem~\ref{thm:theta} (see the proof of Proposition~\ref{prop:T}~\ref{prop:T-conv})).

\end{proof}

 \section*{Acknowledgements}

 This work work was supported by the Agence Nationale de la Recherche (ANR) through the project DYRAQ ANR-17-CE40-0016 and the ICB received additional support through the EUR-EIPHI Graduate School (Grant No. ANR-17-EURE-0002).

\appendix 

\section{Products in $\cK_n$}\label{sect:kappa-prod}

In this appendix we derive in detail the quantitative bounds on $\kappa \star_\ell H_0^{-1} \kappa' $ 
that imply Theorem~\ref{thm:K-star-prod} and are used in the proof of Theorem~\ref{thm:theta}.

Recall the definition~\eqref{eq:rho-def} of $\rho_{n,\lambda}$, $\tilde \rho_{n, \lambda}$ and note that 
\begin{equation}
 \tilde \rho_{n, \lambda}(R,E)=\rho_{n, \lambda}((r_n, \dots, r_1), E).
\end{equation}
We also have the property
\begin{equation}\label{eq:rho-id}
\begin{aligned}
  \rho_{m, 1}(Q, E+\Omega(Q')) \rho_{m', \lambda}(Q', E)&= \rho_{m+m', \lambda}((Q, Q'), E) \\
  \tilde \rho_{m, \lambda}(R,E) \tilde \rho_{m', 1}(R', E+\Omega(R)) &= \tilde \rho_{m+m', \lambda}((R, R'), E).
 \end{aligned}
\end{equation}

We will frequently use an elementary bound on a class of integrals. The proof makes explicit the role of the parameters $\alpha, \gamma$.

\begin{lem}
Assume the hypothesis~\ref{ass:main}.
 Let $s, t\geq 0$ with $s\neq 1+\delta/\gamma$ and $s+t>1+\delta/\gamma$. Then  for $b>0$, $a\geq 0$
 \begin{equation}\label{eq:integral-E}
 \int \frac{|v(\xi)|^2 \ud \xi}{(a+\omega(\xi))^s(b+\omega(\xi))^{t}} \lesssim (a+1)^{-(s-1-\delta/\gamma)_+} b^{-\min\{s+t -1 -\delta/\gamma, t\}},
\end{equation}
\end{lem}
\begin{proof}
 By the hypothesis~\ref{ass:main} the integral is bounded  by
 \begin{equation}
  \eqref{eq:integral-E} \lesssim \int 
  \frac{\ud \xi}{|\xi|^{2\alpha} (a+(c_1+\xi^2)^{\gamma/2})^s(b+|\xi|^{\gamma})^{t}}.
 \end{equation}
For $s>1+\delta/\gamma$, i.e. $2\alpha + s \gamma>d$, we drop $|\xi|^\gamma$ from the second factor in the denominator and obtain
\begin{equation}
 \eqref{eq:integral-E} \lesssim b^{-t} \int \frac{\ud \xi}{|\xi|^{2\alpha} (a+(c_1+\xi^2)^{\gamma/2})^s}
 \lesssim (a+1)^{-s -2\alpha/\gamma +d/\gamma}b^{-t},
\end{equation}
which yields the claim as $-2\alpha/\gamma +d/\gamma=1+\delta/\gamma$.
For $s<1+\delta/\gamma$, we have $2\alpha + s \gamma<d$ and
\begin{equation}
 \eqref{eq:integral-E} \lesssim  \int \frac{\ud \xi}{|\xi|^{2\alpha+s\gamma}(b+|\xi|^{\gamma})^{t} }
 \lesssim b^{-s-t-2\alpha/\gamma+d/\gamma},
\end{equation}
which proves the claim.
\end{proof}

The  following Lemmas propagate bounds on $\kappa\in \cK_n, \kappa'\in \cK_{n'}$ to  $\kappa H_0^{-1} \star_\ell  \kappa' $.
We treat the cases $\ell=\min\{n, n'\}=0$, $0<\ell<\max\{n, n'\}$, and $\ell=n=n'$ separately, starting with the case $\ell=\min\{n, n'\}=0$.

\begin{lem}\label{lem:star ell=0}
 Let $n\in \N$ and $\kappa \in \cK_n$, $\kappa'\in \cK_{0}$. Suppose that for some  $\mu\geq 0$,
 $0\leq \lambda\leq 1$, and $\lambda'\geq 0$ we have the bounds
\begin{align*}
  |\kappa(Q,R,p,E)| & \lesssim E^{-\mu} \min_{s\in [\lambda-1, 1-\lambda]} \rho_{n,\lambda+s}(Q,E)\tilde\rho_{n,\lambda-s}(R,E) \\
  |\kappa'(p,E)| & \lesssim H_0(p,E)E^{-\lambda'}.
\end{align*}
Then with
\begin{align*}
 \sigma&=\min\{\lambda+\lambda', 1\}, \\
 \tau&=\mu+(\lambda+\lambda'-1)_+,
\end{align*}
we have 
\begin{align*}
 |\kappa H_0^{-1} \star_0 \kappa'|(Q,R,p,E) &\lesssim E^{-\tau}
 \min_{s\in [\sigma-1, 1-\sigma]}\rho_{n,\sigma+s}(Q,E)\tilde\rho_{n,\sigma-s}(R,E), \\
 |\kappa' H_0^{-1} \star_0 \kappa|(Q,R,p,E)& \lesssim E^{-\tau}
 \min_{s\in [\sigma-1, 1-\sigma]}\rho_{n,\sigma+s}(Q,E)\tilde\rho_{n,\sigma-s}(R,E).
\end{align*}
\end{lem}
\begin{proof}
 The kernel of $\kappa' H_0^{-1} \star_0 \kappa$ is
 \begin{equation}
  \kappa' H_0^{-1} \star_0 \kappa(Q,R,p,E) = \frac{\kappa'\Big(p+\sum_{\mu=1}^n q_\mu, E + \omega(Q)\Big)\kappa(Q,R,p,E)}{\Omega(p+\sum_{\mu=1}^n q_\mu) + E+\omega(Q)}.
 \end{equation}
We thus have for any $t\in [\lambda-1,1-\lambda]$
\begin{align}
 |\kappa' H_0^{-1} \star_0 \kappa(Q,R,p,E)|&\lesssim E^{-\mu}(E+\omega(Q))^{-\lambda'} \rho_{n,\lambda+t}(Q,E)
 \tilde\rho_{n,\lambda-t}(R,E) \notag \\
 & \lesssim E^{-\mu}\rho_{n,\lambda +2\lambda'+t}(Q,E)
 \tilde\rho_{n,\lambda-t}(R,E).
\end{align}
Choosing $t=s-\sigma+\lambda$ for $s\in [\sigma-1,1-2\lambda+\sigma]\supset [\sigma-1, 1-\sigma]$, this becomes
\begin{align}
 |\kappa' H_0^{-1} \star_0 \kappa(Q,R,p,E)|&\lesssim E^{-\mu}\rho_{n,2\lambda+2\lambda' - \sigma +t}(Q,E)
 \tilde\rho_{n,\sigma-s}(R,E) \notag \\
 & \lesssim E^{-\tau} \rho_{n,\sigma +t}(Q,E)
 \tilde\rho_{n,\sigma-s}(R,E).
\end{align}
The proof for the $\kappa H_0^{-1} \star_0 \kappa'$ is essentially the same, with the roles of $Q,R$ reversed.
\end{proof}

The next Lemma treats the general case of $\kappa\star H_0^{-1} \kappa'$, except for the special case $n=n'=\ell$.

 \begin{lem}\label{lem:star ell<n}
 Let $n, n'\in \N$ and $\kappa \in \cK_n$, $\kappa'\in \cK_{n'}$. Suppose that for some  
 $\mu, \mu'\geq 0$ and  $0\leq \lambda, \lambda'\leq 1$ with $\lambda+\lambda'\neq 1+\delta/\gamma$ we have the bounds
\begin{align*}
  |\kappa(Q,R,p,E)| & \lesssim E^{-\mu} \min_{s\in [\lambda-1, 1-\lambda]} \rho_{n,\lambda+s}(Q,E)\tilde\rho_{n,\lambda-s}(R,E) \\
  |\kappa'(Q',R',p,E)| & \lesssim E^{-\mu'} \min_{s\in [\lambda'-1, 1-\lambda']}  \rho_{n',\lambda'+s}(Q',E)\tilde\rho_{n',\lambda'-s}(R',E).
\end{align*}
Then for all $0\leq \ell\leq \min\{n, n'\}$ with $\ell<\max\{n , n'\}$ and
\begin{align*}
 \sigma&=\min\{\lambda+\lambda'+\ell(1-\delta/\gamma), 1\}, \\
 \tau&=\mu+\mu'+(\lambda+\lambda'+\ell(1-\delta/\gamma)-1)_+,
\end{align*}
we have 
\begin{equation*}
 |\kappa H_0^{-1} \star_\ell \kappa'|(Q,R,p,E) \lesssim E^{-\tau}
 \min_{s\in [\sigma-1, 1-\sigma]}\rho_{n+n'-\ell,\sigma+s}(Q,E)\tilde\rho_{n+n'-\ell,\sigma-s}(R,E).
\end{equation*}
\end{lem}

\begin{proof}

For $\ell=0$, we then have from~\eqref{eq:rho-id} (taking $s=1-\lambda$ and $s'=\lambda'-1$)
 \begin{align}
  &|\kappa \star_0 H_0^{-1}\kappa'|((Q,Q'), (R,R'),p,E) \notag \\
   &\lesssim E^{-\mu-\mu'}\frac{\tilde \rho_{n, 2\lambda-1}(R, E+\omega(Q')) \tilde \rho_{n', 1}(R',E+\omega(R))}{E+\omega(Q')+\omega(R)}  \notag\\
  &\qquad \times \rho_{n, 1}(Q, E+\omega(Q'))\rho_{n', 2\lambda'-1}(Q',E+\omega(R)) \notag \\
  &\leq E^{-\mu-\mu'}\frac{\rho_{n+n',2\lambda'-1}((Q,Q'),E) \tilde \rho_{n+n', 2\lambda-1}((R,R'),E)}{E+\omega(q'_{n'})+\omega(r_1)} \notag\\
  & \leq E^{-\mu-\mu'} \rho_{n+n',2\lambda'+t}((Q,Q'),E) \tilde\rho_{n+n', 2\lambda-t}((R,R'),E),
 \end{align}
 for any $-1\leq t\leq 1$. Setting $s=\lambda'-\lambda+t$ and bounding any inverse powers of $E+\omega(r_1)$, $E+\omega(q_{n+n'})$ in excess of one by inverse powers of $E$ gives the desired inequality. 
 
 Now let $\ell>0$.  In view of~\eqref{eq:kappa-star_ell}, we need to integrate in $\Xi=(\xi_1, \dots , \xi_\ell)$ the quantity
\begin{align}
 &\frac{ \Big|\kappa\Big(Q_1^n, S,p+\sum\limits_{\mu=1}^{n'-\ell} q_{n+\mu}, E+\omega(Q_{n+1}^{n+n'-\ell})\Big) \Big|}  {E+\omega(\Xi)+ \omega(Q_{n+1}^{n+n'-\ell})+\omega(R_1^{n-\ell})} \notag \\
 & \times \Big|\kappa'\Big(U, R_{n-\ell+1}^{n+n'-\ell},p+\sum\limits_{\mu=1}^{n-\ell} r_\mu, E+\omega(R_1^{n-\ell})\Big)\Big|\label{eq:kappa-star-integrand}
\end{align}
evaluated at $S_I=\Xi=U_J$, $S_{I^c}=R_1^{n-\ell}$, $U_{J^c}=Q_{n+1}^{n+n'-\ell}$,
where $I=(i_1, \dots , i_\ell)$ with $1\leq i_1< \dots< i_\ell\leq n$ and $J=(j_1, \dots j_\ell )$ with  pairwise different $j_1, \dots , j_\ell \in \{1, \dots, n'\}$.

We first restrict to $\ell<\min\{n,n'\}$.
As for $\ell=0$, we then use hypothesis with $s=1-\lambda$, $s'=\lambda'-1$ to obtain
\begin{align}
 \eqref{eq:kappa-star-integrand}
 &\leq \frac{\tilde \rho_{n, 2\lambda-1}(S,E+\omega(Q_{n+1}^{n+n'-\ell}))\rho_{n',2\lambda'-1}(U, E+\omega(R_1^{n-\ell}))}{E+\omega(\Xi)+ \omega(Q_{n+1}^{n+n'-\ell})+\omega(R_1^{n-\ell})} \label{eq:rho-ell integrand} \\
 &\qquad \times E^{-\mu-\mu'}\rho_{n, 1}(Q_1^n, E+\omega(Q_{n+1}^{n+n'-\ell}) \tilde \rho_{n',1}(R_{n-\ell+1}^{n+n'-\ell}, E+\omega(R_1^{n-\ell})).\label{eq:rho-non-integrand}
\end{align}
We expand $\rho$, $\tilde \rho$ in~\eqref{eq:rho-ell integrand} using the definition~\eqref{eq:rho-def} and evaluate the variables $S, U$ according to the prescription above.
For any pair with $i_\nu\neq 1$, $j_\nu\neq n'$, $\nu=1,\dots, \ell$, we group the two factors containing $v(s_{i_\nu})=v(\xi_\nu)=v(u_{j_\nu})$ together and drop $\omega(\xi_\nu)$ from all other factors, which gives an upper bound. The integral over $\xi_\nu$ is then given by
\begin{align}
 &\int \hspace{-3.5pt} \frac{|v(\xi_\nu)|^2 \ud\xi_\nu}{(E+\omega(\xi_\nu) + \omega(R_1^{i_\nu -1})+ \omega(Q_{n+1}^{n+n'-\ell}))(E+\omega(\xi_\nu)+\omega(R_1^{n-\ell})+\omega(Q_{j_\nu+1}^{n+n'-\ell}))} \notag\\
 &\stackrel{\eqref{eq:integral-E}}{\lesssim} (E+\omega(r_1)+\omega(q_{n+n'-\ell}))^{\delta/\gamma-1}.
\end{align}
If $\nu=1$ and $i_1=1$, $j_1\neq n'$, we include the factor with $v(r_1)$ (here we use that $\ell<n$ and thus $r_1=s_a$ for some $a>1$) before dropping $\omega(\xi_1)$, which gives an upper bound on the $\xi_1$-integral by
\begin{align}
 &| v(r_1)| \int  \frac{|v(\xi_1)|^2 \ud\xi_1}{(E+\omega(\xi_1))^\lambda(E +\omega(\xi_1) + \omega(r_1)+\omega(q_{n+n'-\ell}))^2} \notag\\
 &\stackrel{\eqref{eq:integral-E}}{\lesssim} \frac{| v(r_1)|}{(E+\omega(r_1)+\omega(q_{n+n'-\ell}))^{\lambda +1-\delta/\gamma}},
\end{align}
since $\lambda\leq 1<1+\delta/\gamma$.
If for some $\nu\in \{1, \dots, \ell\}$, $j_\nu=n'$ and $i_\nu\neq 1$ the argument is the same with $v(r_1)$ replaced by $v(q_{n+n'-\ell})$ and $\xi_1$ by $\xi_\nu$.

 This gives us the inequality
\begin{align}
 &\int \eqref{eq:rho-ell integrand}\Bigg\vert{\substack{S_I=\Xi \\ S_{I^c}=R_1^{n-\ell} \\ U_J=\Xi \\ U_{J^c}=Q_{n+1}^{n+n'-\ell}}} \ud \Xi \notag \\
 &\lesssim  \frac{\rho_{n'-\ell, 2\lambda'-1}(Q_{n+1}^{n+n'-\ell},E)\tilde \rho_{n-\ell,2\lambda-1 }(R_1^{n-\ell}, E)}{(E+\omega(r_1)+\omega(q_{n+n'-\ell}))^{1+\ell(1-\delta/\gamma)}}
\end{align}
 Combining this with~\eqref{eq:rho-non-integrand} and splitting the denominator as in the case $\ell=0$ gives a bound by
\begin{align}
 E^{-\tau} \rho_{n+n'-\ell,\sigma + s}(Q,E)\tilde\rho_{n+n'-\ell, \sigma-s}(R,E),
\end{align}
with $2\lambda'-1-\sigma \leq s\leq \sigma+1-2\lambda$, which is a weaker condition than $\sigma-1\leq s\leq 1-\sigma$. This proves required bound for the terms with $j_\nu\neq n'$.

If $I, J$ are such that $i_1=1$ and $j_1=n'$ we include both the factor with $v(r_1)$ and $v(q_{n+n'-\ell})$ in the $\xi_1$-integral, leading to 
\begin{align}
 & |v(r_1)||v(q_{n+n'-\ell})| \int \frac{|v(\xi_1)|^2 \ud \xi_1 }{(E+\omega(\xi_1))^{\lambda+\lambda'}(E +\omega(\xi_1) + \omega(r_1)+\omega(q_{n+n'-\ell}))^2} \notag \\
 & \stackrel{\eqref{eq:integral-E}}{\lesssim}  \frac{|v(r_1)||v(q_{n+n'-\ell})|}{(E+\omega(r_1)+\omega(q_{n+n'-\ell}))^{\min\{\lambda+\lambda'+1-\delta/\gamma,2\}}} E^{-(\lambda+\lambda'-1-\delta/\gamma)_+} \notag \\
 &\lesssim 
   \frac{|v(r_1)||v(q_{n+n'-\ell})|}{(E+\omega(r_1)+\omega(q_{n+n'-\ell}))^{\sigma - (\ell-1)(1-\delta/\gamma)}} E^{-(\lambda + \lambda'+\ell(1-\delta/\gamma)-1)_+}
\end{align}
where we used that $1+\delta/\gamma \neq \lambda + \lambda'$. From here we conclude as before, and this proves the claim for  $\ell<\min\{n, n'\}$.

The remaining case is $\ell=\min\{n, n'\}$. Let $\ell=n$, $\ell<n'$. We then use the hypothesis differently, keeping the freedom of choosing the value of $s'\in [\lambda'-1, 1-\lambda']$. Instead of~\eqref{eq:rho-ell integrand},~\eqref{eq:rho-non-integrand}, this gives for the case at hand
\begin{align}
 \eqref{eq:kappa-star-integrand}&=\frac{\Big|\kappa\big(Q_1^n, \Xi,p+\sum_{\mu=1}^{n'-\ell} q_{n+\mu}, E+\omega(Q_{n+1}^{n+n'-\ell})\big) \kappa'\big(U, R,p, E)\Big|}{E+\omega(\Xi)+ \omega(Q_{n+1}^{n+n'-\ell})} \notag\\
 &\lesssim  \frac{\tilde\rho_{n, 2\lambda-1}(\Xi,E+\omega(Q_{n+1}^{n+n'-\ell}))\rho_{n',\lambda'+s'}(U, E)}{E+\omega(\Xi)+ \omega(Q_{n+1}^{n+n'-\ell})} \label{eq:rho n=ell int} \\
 &\qquad \times E^{-\mu-\mu'} \rho_{n, 1}(Q_1^n, E+\omega(Q_{n+1}^{n+n'-\ell})) \tilde \rho_{n',\lambda'-s'}(R, E)).
 \label{eq:rho n=ell non-int}
\end{align}
The integral over $\xi_1=\xi_{i_1}$ is then bounded using the denominator in~\eqref{eq:rho n=ell int}, which gives for $j_1\neq n'$
\begin{align}
 &\int  \frac{|v(\xi_1)|^2 \ud\xi_1}{(E +\omega(\xi_1) +\Omega(Q_{n+1}^{n+n'-\ell}))^{1+\lambda}(E +\omega(\xi_1) +\Omega(Q_{n+j_1+1}^{n+n'-\ell}))} \notag \\
 &\stackrel{\eqref{eq:integral-E}}{\lesssim}  (E+\omega(q_{n+n'-\ell}))^{-\lambda-1+\delta/\gamma}.
\end{align}
If $j_1=n'$, then we additionally include the factor with $v(q_{n+n'-\ell})$ as we did for $\ell<\min\{n, n'\}$ and obtain (with $(1+\lambda'-s')/2\leq 1$)
\begin{align}
 &|v(q_{n+n'-\ell})| \int  \frac{|v(\xi_1)|^2 \ud\xi_1}{(E +\omega(\xi_1) +\omega(q_{n+n'-\ell}))^{2+\lambda}(E +\omega(\xi_1))^{(1+\lambda'-s')/2}} \notag \\
 &\lesssim \frac{|v(q_{n+n'-\ell})|}{(E+\omega(q_{n+n'-\ell}))^{1-\delta/\gamma + (1+2\lambda+\lambda'-s')/2}}.
\end{align}
This gives us a bound on the integral by
\begin{align}
 \int \eqref{eq:rho n=ell int}  \bigg|{\substack{\scriptscriptstyle U_J=\Xi \\ U_{J^c}=Q_{n+1}^{n+n'-\ell}}}
 \; \ud \Xi
 \lesssim \rho_{n'-\ell,2\lambda+\lambda'+2\ell(1-\delta/\gamma)+ t}(Q_{n+1}^{n+n'-\ell}, E).  
\end{align}
Combining with~\eqref{eq:rho n=ell non-int} and setting $s=s'+\sigma-\lambda'$ (with the resulting restriction $ s\in [\sigma-1, 1+\sigma-2\lambda']\supset [\sigma-1,1-\sigma]$)
gives the claim.
 \end{proof}

We now turn to the remaining case $\ell=n=m$.

\begin{lem}\label{lem:star ell=n}
Let $n\in \N$ and $\kappa, \kappa' \in \cK_n$. Suppose that for some  
 $\mu, \mu'\geq 0$ and  $0\leq \lambda, \lambda'\leq 1$  we have the bounds
\begin{align*}
  |\kappa(Q,R,p,E)| & \lesssim E^{-\mu} \min_{s\in [\lambda-1, 1-\lambda]} \rho_{n,\lambda + s}(Q,E)\tilde\rho_{n,\lambda-s}(R,E) \\
   |\kappa'(Q,R,p,E)| & \lesssim E^{-\mu'} \min_{s\in [\lambda'-1, 1-\lambda']} \rho_{n,\lambda' + s}(Q,E)\tilde\rho_{n,\lambda'-s}(R,E)
\end{align*}
for $E\geq 1$.
Then for all $0\leq \sigma \leq 1$ satisfying
\begin{align*}
 \max\{\lambda, \lambda'\}\leq \sigma&<\lambda+\lambda'+n(1-\delta/\gamma)
 \end{align*}
 and
 \begin{align*}
 \tau&=\mu+\mu'+(\lambda+\lambda'+n(1-\delta/\gamma)-1)_+
\end{align*}
we have for $E\geq 1$
\begin{equation*}
 |\kappa \star_n H_0^{-1}\kappa'|(Q,R,p,E) \lesssim E^{-\tau}\min_{s\in [\sigma-1, 1-\sigma]}  \rho_{n,\sigma_\eps+s}(Q,E)\tilde\rho_{n,\sigma-s}(R,E).
\end{equation*}
\end{lem}
 
 \begin{proof}
  This is the case $n=\ell=n'$ of the previous lemma and we adopt the notation from there.
In the present case we have $i_\nu=\nu$, and $j_1, \dots, j_n$ is just a permutation of $1, \dots, n$. The integral then simplifies to
\begin{align}
 \int \frac{\Big|\kappa\big(Q, \Xi, E)\big) \kappa'\big(\Xi_J, R, E)\Big|}{E+\Omega(\Xi)} \ud \Xi 
 &\lesssim E^{-\mu-\mu'} \rho_{n, \lambda + t}(Q, E) \tilde \rho_{n,\lambda'-t'}(R,E)\notag \\
 & \times \int \frac{\tilde \rho_{n, \lambda-t}(\Xi, E)\rho_{n, \lambda'+t'}(\Xi_J, E)}{E+\omega(\Xi)} \ud \Xi  \label{eq:int ell=n=n'},
\end{align}
where $\Xi_J=(\xi_{j_1}, \dots, \xi_{j_n})$ are the permuted variables.
Hence we only need to prove that the integral~\eqref{eq:int ell=n=n'} is bounded by $E^{-(\lambda+\lambda'+n(1-\delta/\gamma)-1)_+}$ for appropriate choices of $t, t'$.

Let us first consider the case $j_1=n$. Then the integral is bounded by
\begin{align}
 \eqref{eq:int ell=n=n'} &\leq \int \frac{|v(\xi_1)|^2}{(E+\omega(\xi_1))^{2+(\lambda- t+\lambda'+t')/2}} \prod_{j=2}^n \frac{|v(\xi_j)|^2}{(E+\omega(\xi_j))(E+\omega(\Xi_1^j))} \ud \Xi \notag \\
 &\stackrel{\eqref{eq:integral-E}}{\lesssim}  \int  \frac{|v(\xi_1)|^2\ud \xi_1}{(E+\omega(\xi_1))^{2+(\lambda - t+\lambda'+t')/2+(n-1)(1-\delta/\gamma)}}.
\end{align}
For the final integral to be finite, we need
\begin{align}
 2+(\lambda + \lambda'- t+t')/2+(n-1)(1-\delta/\gamma)>1+\delta/\gamma .
\end{align}
Now let $\sigma \leq 1$ as in the statement,
and set for $\sigma-1\leq s\leq 1-\sigma$
\begin{equation}
t=s+\sigma-\lambda \,, \qquad t'=s - \sigma + \lambda'.
\end{equation}
These choices are admissible since $\sigma\geq \lambda, \lambda'$. Because $\sigma < \lambda + \lambda' + n(1-\delta/\gamma)$, we have 
\begin{align}
  \eqref{eq:int ell=n=n'}&\lesssim \int  \frac{|v(\xi)|^2\ud \xi}{(E+\omega(\xi))^{2+(\lambda +\lambda'+(n-1)(1-\delta/\gamma)-\sigma)}}
  \stackrel{\eqref{eq:integral-E}}{\lesssim}  E^{-(\lambda +\lambda'+n(1-\delta/\gamma)-\sigma)}.
\end{align}
As $\tau -\mu-\mu'\leq  (\lambda +\lambda'+n(1-\delta/\gamma)-\sigma)$, this yields
\begin{align}
 |\kappa \star_n H_0^{-1}\kappa'|(Q,R,p,E) \lesssim E^{-\tau} \min_{s\in [\sigma-1, 1-\sigma]} \rho_{n, \sigma+s}(Q,E)\tilde\rho_{n, \sigma-s}(R,E).
\end{align}
This proves the claim for the case $j_1=n$.

The case $j_1\neq n$ arises only for $n\geq 2$. We choose $t,t'$ as before, and we then group the denominator in~\eqref{eq:int ell=n=n'} with the $\xi_{j_1}$-integral. This is then bounded by
\begin{align}
 &\int \frac{|v(\xi_{j_1})|^2\ud \xi_{j_1}}{(E+\omega(\xi_{j_1}))^{(1+2\lambda'+s-\sigma)/2}(E+\omega(\xi_1) + \omega(\xi_{j_1}))^2} \notag \\
 &\stackrel{\eqref{eq:integral-E}}{\lesssim} (E+\omega(\xi_1))^{-(1-\delta/\gamma)-(1+2\lambda'+s-\sigma)/2},
\end{align}
since $2\lambda'+s-\sigma \leq \lambda'+s \leq 1 $.
Treating the integrals over $\xi_\nu$, $\nu=2, \dots, n$, $\nu\neq j_1$, in the same way as before, we are left with the $\xi_1$-integral
\begin{align}
  \eqref{eq:int ell=n=n'}&\lesssim \int  \frac{|v(\xi)|^2\ud \xi}{(E+\omega(\xi))^{2+(\lambda +\lambda'+(n-1)(1-\delta/\gamma)-\sigma)}}.
\end{align}
This proves the claim by the same argument as for $j_1=1$. 
 \end{proof}

 In the proof of Theorem~\ref{thm:theta} we additionally need bounds on $a(v)H_0^{-1}\star_{\ell} \kappa H_0^{-1}\star_{\ell'} a^*(v)$.
 
\begin{lem}\label{lem:a-star}
 Let $n\in \N$, $\kappa\in \cK_n$ and $\ell, \ell'\in \{0,1\}$ with $\ell+\ell'\leq n$.
 Suppose that for some $\mu\geq 0$ and $0\leq \lambda \leq 1$ we have
 \begin{equation*}
  |\kappa(Q,R,p,E)| \lesssim E^{-\mu} \min_{s\in [\lambda-1, 1-\lambda]} \rho_{n,\lambda+s}(Q,E)\tilde\rho_{n,\lambda-s}(R,E).
 \end{equation*}
Then
\begin{equation*}
\kappa_{\ell, \ell'}:= \Big(a(v)H_0^{-1}\star_{\ell} \kappa  \Big) H_0^{-1}\star_{\ell'} a^*(v) \in \cK_{n+1-\ell-\ell'},
\end{equation*}
and for 
\begin{align*}
 \sigma&= \min\{ \lambda + (\ell+\ell')(1-\delta/\gamma),1\} \\
 \tau &= \mu + (\lambda + (\ell+\ell')(1-\delta/\gamma)-1)_+
\end{align*}
we have
\begin{equation*}
 |\kappa_{\ell, \ell'}(Q,R,p,E)|\lesssim E^{-\tau} \min_{s\in [\sigma-1, 1-\sigma]} \rho_{n,\sigma+s}(Q,E)\tilde\rho_{n,\sigma-s}(R,E).
\end{equation*}
\end{lem}

 \begin{proof}
   The kernel of $\kappa_{0,0}$ is
  \begin{align}
   &\kappa_{0,0}(Q,R,p,E) \notag \\
   &= \frac{v(r_1) v(q_{n+1}) \kappa\Big(Q_1^n,R_1^n, p+r_1+q_{n+1}, E+ \omega(r_1)+\omega(q_{n+1})\Big)}
   {(E+ \omega(Q)+\omega(r_1))(E+\omega(R)+\omega(q_{n+1}))}.
   \label{eq:kappa_00-def}
  \end{align}
Using the hypothesis, it thus satisfies
\begin{align}
&|\kappa_{0,0}(Q,R,p,E)| \notag \\
%
%
&\lesssim \frac{|v(r_1)| |v(q_{n+1})| \rho_{n, 1}(Q_1^n, E+\omega(q_{n+1}))\tilde \rho_{n, 2\lambda-1}(R_2^{n+1}, E+ \omega(r_1)+\omega(q_{n+1})) }{(E+ \omega(q_{n+1})+\omega(r_1))^{1+\mu+\lambda}(E+\omega(R))^{1-\lambda}} \notag \\
&= \frac{ \rho_{n+1, -1}(Q_1, E))  \tilde \rho_{n+1, -1}(R, E) }{(E+ \omega(q_{n+1})+\omega(r_1))^{1+\lambda+\mu}} \notag \\
&\leq E^{-\mu} \rho_{n+1, \lambda + s}(Q, E))  \tilde \rho_{n+1, \lambda-s}(R, E)
 \end{align}
for any $-1-\lambda\leq s \leq 1+\lambda$. 

For the case $\ell+\ell'=1$, we give the details only for $\kappa_{0,1}$ (the proof for $\kappa_{1,0}$ is similar, but the term corresponding to $i=1$ below does not occur).
The kernel of the first parenthesis, where there is no contraction, satisfies
\begin{align}
&\left| a(v)H_0^{-1} \star_0 \kappa (Q,R,p,E) \right| \notag\\
 &\lesssim E^{-\mu} \frac{|v(r_1)| \rho_{n, \lambda + s}(Q, E) \tilde\rho_{n, \lambda - s}(R_2^{n+1}, E+\omega(r_1))}{E+\omega(Q)+\omega(r_1)}.
\end{align}
Distinguishing the contraction with the first variable from the remaining ones, where we take $s=\lambda-1$, we obtain
\begin{align}
 &|\kappa_{0,1} (Q,R,p,E)| \notag \\
 &=\left| \sum_{i=1}^{n+1} \int  \frac{\big(a(v)H_0^{-1} \star_0 \kappa\big)(Q,S,p,E) v(\xi)}{E+\omega(\xi) + \omega(R)}\bigg\vert\substack{s_i=\xi \\ S_{\{i\}^c}=R}\, \ud \xi \right| \notag \\
 &\lesssim  E^{-\mu} \rho_{n, \lambda+s}(Q, E)\tilde \rho_{n, \lambda-s}(R,E) \int \frac{|v(\xi)|^2 \ud \xi}{(E+\omega(\xi) + \omega(Q))(E+\omega(\xi) + \omega(R))}  \notag \\
 &\qquad  + \frac{| v(r_1)|   \tilde \rho_{n-1, 1}(R_2^n, E+\omega(r_1)) \rho_{n, 2\lambda-1}(Q, E)}{E+\omega(Q)+\omega(r_1)} \notag \\ 
 &\qquad\qquad  \times \sum_{i=2}^{n+1} \int \frac{|v(\xi)|^2 \ud \xi}{(E+\omega(\xi) + \omega(R_1^{i-1}))(E+\omega(\xi) + \omega(R))} \notag \\
 &\stackrel{\eqref{eq:integral-E}}{\lesssim} E^{-\mu} \frac{\rho_{n, \lambda+s}(Q, E)\tilde \rho_{n,\lambda -s}(R,E)}{ (E+\omega(r_1)+\omega(q_{n}))^{1-\delta/\gamma}} \notag \\
 &\qquad + nE^{-\mu}  \frac{\tilde \rho_{n, -1 +2 (1-\delta/\gamma)}(R, E) \rho_{n, 2\lambda - 1}(Q, E)}{E+\omega(r_1)+\omega(q_n)} \notag \\
 &\lesssim E^{-\mu - (\lambda -\delta/\gamma)_+}\rho_{n, \sigma+s}(Q, E)\tilde \rho_{n, \sigma-s}(R,E)
\end{align}
for $\sigma-1\leq s \leq 1-\sigma$ (in fact, the range can be chosen larger here).
This shows the bound as claimed.

For $\ell=\ell'=1$ recall that we suppose that $n\geq 2$. The kernel we need to bound is
\begin{align}
 &|\kappa_{1,1}(Q,R,p,E)| \notag \\
 &= \left|\sum_{i, j=1}^n \int\frac{v(\xi_1) }{E+\omega(\xi_1)+ \omega(Q)} \frac{v(\xi_2)}{E+\omega(\xi_2) + \omega(R)}  \kappa(U,S,p,E)\bigg\vert\substack{u_j=\xi_1 \\ U_{\{j\}^c}=Q \\ s_i=\xi_2 \\ S_{\{i\}^c}=R }
 \, \ud \xi_1 \ud \xi_2\right| \notag \\
 &\lesssim  E^{-\mu}\bigg(\sum_{j=1}^n \int \frac{|v(\xi_1)| \rho_{n, \lambda+s}(U, E) }{E+\omega(\xi_1) + \omega(q_{n-1})}
 \bigg\vert\substack{u_j=\xi_1 \\ U_{\{j\}^c}=Q }\,\ud \xi_1\bigg) \label{eq:tau_11 int-Q} \\
&\qquad \times \bigg(\sum_{i=1}^n \int \frac{|v(\xi_2)| \tilde\rho_{n, \lambda-s}(S, E) }{E+\omega(\xi_2) + \omega(r_1)}
 \bigg\vert\substack{s_i=\xi_2 \\ S_{\{i\}^c}=R }
 \, \ud \xi_2\bigg) \label{eq:tau_11 int-R}
\end{align}
To bound the integral~\eqref{eq:tau_11 int-R}, we expand $\tilde \rho$ as a product of fractions with numerator $|v(s_\nu)|$ using its definition.
For $i\geq 2$, we drop $\omega(\xi_2)$ from the denominators of all factors except the one of $v(\xi_2)=v(s_i)$, which gives the bound
\begin{align}
 \int \frac{|v(\xi)| \tilde\rho_{n, \lambda-s}(S, E) }{E+\omega(\xi) + \omega(r_1)}
 \bigg\vert\substack{s_i=\xi \\ S_{\{i\}^c}=R }
 \, \ud \xi 
 &\leq \tilde\rho_{n-1, \lambda-s}(R, E) \int \frac{|v(\xi)|^2}{(E+\omega(\xi) + \omega(r_1))^2} \notag \\
 &\stackrel{\eqref{eq:integral-E}}{\lesssim} \tilde\rho_{n-1, \lambda-s+2(1-\delta/\gamma)}(R, E).
\end{align}
For $i=1$ we do not drop $\omega(\xi)$ in the factor of $v(r_1)=v(s_2)$. 
This leads to (keeping in mind that $1+\lambda-s \leq 2$)
\begin{align}
 &\int \frac{|v(\xi)| \tilde\rho_{n, \lambda-s}((\xi, R), E) }{E+\omega(\xi) + \omega(r_1)}
  \, \ud \xi \notag\\
 & \leq \tilde\rho_{n-2,1}(R_2^{n-1}, E+\omega(r_1))  \int \frac{|v(r_1)| |v(\xi)|^2 \ud \xi}{(E+\omega(\xi))^{(1+\lambda-s)/2}(E+\omega(\xi) + \omega(r_1))^2} \notag \\
 &\stackrel{\eqref{eq:integral-E}}{\lesssim} \tilde \rho_{n-1, \lambda-s+2(1-\delta/\gamma)}(R,E)\leq E^{-(\lambda-\delta/\gamma)_+} \tilde \rho_{n-1, \sigma-s}(R,E).
\end{align}
Arguing in the same way for the other integral~\eqref{eq:tau_11 int-Q} proves the claim.
\end{proof}

\section{Operator bounds}

We first give a well known Lemma on the boundedness of $a(v)\ud \Gamma(\omega)^{-s}$. 

 \begin{lem}\label{lem:a-omega}
 For $s>\tfrac12(1+\delta/\gamma)$
 \begin{equation*}
  \| a(v) d\Gamma(\omega)^{-s} \| \leq \|v\omega^{-s}\|_{L^2}.
 \end{equation*}
\end{lem}

\begin{proof}
 Let $n\in \N_0$ and $\Psi \in \mathcal{F}^{(n+1)}$. Then, using Cauchy-Schwarz inequality, the symmetry of $\Psi$, and the fact that $2s\geq1$, we have
 \begin{align}
 &\|a(v) d\Gamma(\omega)^{-s} \Psi \|^2_{\mathcal{F}^{(n)}} \notag \\
 &=(n+1) \int_{\R^{dn}}  \bigg| \int_{\R^d} \frac{v(\xi)}{\omega^s(\xi)} \frac{\omega^s(\xi)\Psi(K, \xi)}{(\omega(\xi)+\sum_{j=1}^{n} \omega(k_j))^{s}} d \xi \bigg|^2 d K\notag \\
 &\leq \|v\omega^{-s}\|^2 (n+1) \int_{\R^{(n+1)d}} \omega^{2s}(k_{n+1}) \frac{|\Psi(K)|^2}{\big(\sum_{j=1}^{n+1} \omega(k_j)\big)^{2s}} dK  \notag\\
 &= \|v\omega^{-s}\|^2 \int_{\R^{(n+1)d}} \frac{\sum_{\ell=1}^{n+1} \omega^{2s}(k_{\ell})}{(\sum_{j=1}^{n+1} \omega(k_j))^{2s}} |\Psi(K)|^2 dK  \notag \\
 %
 %
 &\leq \|v\omega^{-s} \|^2_{L^2}\|\Psi\|^2_{\mathcal{F}^{(n+1)}}.
 \end{align}
 This proves the claim.
\end{proof}

The following Lemmas provide bounds on elements of $\cK_{n, \lambda}$, as operators respectively quadratic forms. Similar bounds (with $n=1$) appear in the literature on contact interactions~\cite{DeFiTe1994,MoSe17, GrLi18}.

\begin{lem}\label{lem:K_n-op}
 Let $\kappa \in \cK_n$ with kernel satisfying
 \begin{align*}
&|\kappa(Q,R,p,E)|  
\lesssim \rho_{n,1}(Q,E)\tilde\rho_{n,2\lambda-1}(R,E),
 \end{align*}
 for some $0\leq \lambda\leq 1$.
 Then for any non-negative $s>1-\lambda-\tfrac12 n(1-\delta/\gamma)$, the formula~\eqref{eq:kernel-a*a} defines a bounded operator
 \begin{equation*}
  \kappa:D\Big(\ud\Gamma(\omega)^{s}\Big)\to \cF. 
 \end{equation*}
 \end{lem}
 
 \begin{proof}
  Let $\Phi, \Psi \in \cF$ be finite linear combinations of compactly supported functions in $L^2(\R^{dn})$ and note that the set of such elements is dense in $\cF$. We have for any function $h$ on $\R^d$
 \begin{align}
  &|\langle \Phi, \kappa \Psi \rangle| \notag\\
  &\leq \int\limits_{\R^{dn}\times \R^{dn}}  \left|\left\langle \Big(\prod_{i=1}^n a_{q_i}\Big) \Phi , \kappa(Q,R,\ud \Gamma(k)-P, \ud \Gamma(\omega)+E_0) \Big(\prod_{i=1}^n a_{r_i}\Big) \Psi\right\rangle\right|\ud Q \ud R \notag\\
  &\lesssim \int\limits_{\R^{dn}\times \R^{dn}}
  \left\| \rho_{n,1}(Q, \ud \Gamma(\omega))\Big(\prod_{i=1}^n a_{q_i}\Big) \Phi \right\|
  \left\| \tilde\rho_{n,2\lambda-1}(R, \ud \Gamma(\omega))\Big(\prod_{i=1}^n a_{r_i}\Big) \Psi \right\|
  \ud Q \ud R \notag \\
  &\leq \Bigg(\int\limits \Big(\prod_{i=1}^n \frac{h(q_i)}{h(r_i)}\Big)
  \left\langle \Big(\prod_{i=1}^n a_{q_i}\Big) \Phi, \rho_{n,1}(Q, \ud \Gamma(\omega))^2 \Big(\prod_{i=1}^n a_{q_i}\Big) \Phi \right\rangle \ud R \ud Q  \Bigg)^{1/2} \label{eq:K-norm bound phi} \\
  & \quad\times \Bigg(\int \Big(\prod_{i=1}^n \frac{h(r_i)}{h(q_i)}\Big)
  \left\langle \Big(\prod_{i=1}^n a_{r_i}\Big) \Psi, \tilde\rho_{n,2\lambda-1}(R, \ud \Gamma(\omega))^2 \Big(\prod_{i=1}^n a_{r_i}\Big) \Psi \right\rangle \ud Q \ud R  \Bigg)^{1/2}.\label{eq:K-norm bound psi}
\end{align}
We choose $h(q)=\omega(q)^t/|v(q)|^2$ with $t=1+\delta/\gamma + \eps/n$, where $\eps>0$ is such that $t\leq 2$.
Then $\int \frac{\ud q}{h(q)}<\infty$ and, using the pull-through formula~\eqref{eq:pull-gen}, 
\begin{align}
 \eqref{eq:K-norm bound psi}
 &\lesssim\left\langle  \Psi,\int\limits_{\R^{dn}}\Big(\prod_{i=1}^n h(r_i) a_{r_i}^*\Big) \tilde\rho_{n,2\lambda-1}(R, \ud \Gamma(\omega))^2 \Big(\prod_{i=1}^n a_{r_i}\Big) \Psi \ud R \right\rangle^{1/2} \notag \\
 &= \Bigg\langle  \Psi,\int\limits_{\R^{dn}}\Big(\prod_{i=1}^n a_{r_i}^*\Big)
 \frac{\omega(r_1)^t}{(\ud \Gamma(\omega) +\omega(r_1))^{2\lambda}}\notag \\
 &\qquad \times
 \bigg(\prod_{j=2}^{n}\frac{\omega(r_j)^t}{(\ud \Gamma(\omega) + \omega(R_1^j))^2} \bigg) \Big(\prod_{i=1}^n a_{r_i}\Big) \Psi \ud Q \Bigg\rangle^{1/2} \notag \\
 &= \Bigg\langle  \Psi,\int\limits_{\R^{dn}}\Big(\prod_{i=2}^{n} a_{r_i}^*\Big)
 \omega(r_1)^t a^*_{r_1} a_{r_1}\ud\Gamma(\omega)^{-2\lambda}
 \notag \\
 &\qquad \times \bigg(\prod_{j=2}^{n}\frac{\omega(r_j)^t}{(\ud \Gamma(\omega) + \omega(R_2^j))^2} \bigg) \Big(\prod_{i=2}^{n} a_{r_i}\Big) \Psi \ud Q \Bigg\rangle^{1/2} \notag .
 \end{align}

 Now since $t\geq 1$,
\begin{align}
 \int_{\R^d}\omega(r_i)^t a^*_{r_1} a_{r_1}\ud\Gamma(\omega)^{-2\lambda} \ud q_n = \ud \Gamma(\omega^{t})\ud\Gamma(\omega)^{-2\lambda} \leq \ud\Gamma(\omega)^{t-2\lambda}.
\end{align}
Assume first that $t\geq 2\lambda$. Then $\ud\Gamma(\omega)^{t-2\lambda} \leq (\ud\Gamma(\omega)+ \omega(r_2))^{t-2\lambda}$, and we can iterate this argument to obtain 
\begin{align}
 &\int\limits_{\R^{d\nu}}\Big(\prod_{i=1}^{\nu} a_{r_i}^*\Big)\frac{\omega(r_1)^t}{(\ud \Gamma(\omega) +\omega(r_1))^{2\lambda}}
 \left(\prod_{j=2}^{n}\frac{\omega(r_j)^t}{(\ud \Gamma(\omega) + \omega(R_1^j))^2} \right) \Big(\prod_{i=1}^\nu a_{r_i}\Big) \notag \\
 &\leq (\ud\Gamma(\omega)+\omega(r_{\nu+1}))^{t\nu-2(\nu-1)-2\lambda}\left(\prod_{j=\nu+1}^{n}\frac{\omega(r_j)^t}{(\ud \Gamma(\omega) + \omega(R_{\nu+1}^j))^2} \right),
\end{align}
as long as 
\begin{equation}
 t\nu-2(\nu-1)-2\lambda = \nu(\delta/\gamma-1)+2-2\lambda + \frac{\eps \nu}{n}\geq 0.
\end{equation}
Now if $\lambda+\tfrac12 n(1-\delta/\gamma)\leq 1$ this holds true up to $\nu=n$ and we obtain
\begin{equation}
 \eqref{eq:K-norm bound psi} \lesssim \left\langle \Psi, \ud\Gamma(\omega)^{2-2\lambda+n(\delta/\gamma-1)+\eps)} \Psi\right\rangle^{1/2} = \|\ud \Gamma(\omega)^{1-\lambda-n(1-\delta/\gamma)/2 +\eps/2} \Psi \|.
\end{equation}
If $\lambda+\tfrac12n(1-\delta/\gamma)> 1$, let $\nu_0$ be the smallest $\nu \geq 1$ (which exists for small enough $\eps$) such that $t\nu-2\nu+2-\lambda +\eps \nu/n\leq 0$. Then we proceed as before, but bound the negative power of $\ud\Gamma(\omega)$ by a constant, which gives
\begin{align}
&\int\limits_{\R^{d\nu_0}}\Big(\prod_{i=1}^{\nu_0} a_{r_i}^*\Big)\frac{\omega(r_1)^t}{(\ud \Gamma(\omega) +\omega(r_1))^{2\lambda}}
 \left(\prod_{j=2}^{n}\frac{\omega(r_j)^t}{(\ud \Gamma(\omega) + \omega(R_1^j))^2} \right) \Big(\prod_{i=1}^{\nu_0} a_{r_i}\Big) \notag \\
 &\lesssim \prod_{j=\nu_0+1}^{n}\frac{\omega(r_j)^t}{(\ud \Gamma(\omega) + \omega(R_{\nu_0+1}^j))^2} ,
\end{align}
These remaining factors lead to a bounded operator by the same reasoning, because $t\leq 2$, so we obtain in this case
\begin{equation}
 \eqref{eq:K-norm bound psi} \lesssim \|\Psi\|.
\end{equation}
By the same argument, up to renaming of $R,Q$, we also have
\begin{equation}
 \eqref{eq:K-norm bound phi} \lesssim \|\Phi\|,
\end{equation}
and thus
\begin{align}
 &|\langle \Phi, \kappa \Psi \rangle| \lesssim \|\Phi\| \|\ud \Gamma(\omega)^{(1-\lambda-n(1-\delta/\gamma)/2 +\eps/2)_+} \Psi \|,
\end{align}
which proves the claim. 
 \end{proof}

\begin{lem}\label{lem:K_n-form}
 Let $\kappa\in \cK_n$ with kernel satisfying
 \begin{align*}
&|\kappa(Q,R,p,E)|  
\lesssim \rho_{n,\lambda}(Q,E)\tilde\rho_{n,\lambda}(R,E),
 \end{align*}
 for some $0\leq \lambda\leq 1$.
 Then for any non-negative $s>1-\lambda-n(1-\delta/\gamma)$, $\kappa$ defines a quadratic form on $D\Big(\ud\Gamma(\omega)^{s/2}\Big)$ satisfying
 \begin{equation*}
 |\langle \Phi, \kappa \Psi \rangle | \lesssim \| \ud \Gamma(\omega)^{s}\Phi \| \| \ud \Gamma(\omega)^{s}\Psi \|. 
 \end{equation*}
In particular, if $\lambda+n(1-\delta/\gamma)>1$ then $\kappa$ defines a bounded operator on $\cF$.
 \end{lem}

\begin{proof}
As in the proof of Lemma~\ref{lem:K_n-op}, we take $\Phi, \Psi$ as finite combinations of compactly supported functions and obtain 
 \begin{align}
  &|\langle \Phi, \kappa \Psi \rangle|^2 \\
  &\leq \int\limits_{\R^{dn}\times \R^{dn}} \Big(\prod_{i=1}^n \frac{h(q_i)}{h(r_i)}\Big)
  \left\langle \Big(\prod_{i=1}^n a_{q_i}\Big) \Phi, \rho_{n,\lambda}(Q, \ud \Gamma(\omega))^2 \Big(\prod_{i=1}^n a_{q_i}\Big) \Phi \right\rangle \ud R \ud Q   \notag \\
  &\qquad \times \int\limits_{\R^{dn}\times \R^{dn}} \Big(\prod_{i=1}^n \frac{h(r_i)}{h(q_i)}\Big)
  \left\langle \Big(\prod_{i=1}^n a_{r_i}\Big) \Psi, \tilde\rho_{n,\lambda}(R, \ud \Gamma(\omega))^2 \Big(\prod_{i=1}^n a_{r_i}\Big) \Psi \right\rangle \ud Q \ud R  .\notag
\end{align}
With the identical choice for $h$, we conclude using the arguments of Lemma~\ref{lem:K_n-form} that
\begin{align}
 &|\langle \Phi, \kappa \Psi \rangle| \lesssim \|\ud \Gamma(\omega)^{(1-\lambda-n(1-\delta/\gamma) +\eps)/2} \Phi \|  \|\ud \Gamma(\omega)^{(1-\lambda-n(1-\delta/\gamma) +\eps)/2} \Psi \|,
\end{align}
if $\lambda + n(1-\delta/\gamma) \leq 1$, and
\begin{equation}
 |\langle \Phi, \kappa \Psi \rangle| \lesssim \|\Phi\| \|\Psi\|
\end{equation}
if $\lambda + n(1-\delta/\gamma)>1$. This proves the claim. 
\end{proof}


\end{document}